\documentclass[conference,letterpaper]{IEEEtran}
\usepackage[utf8]{inputenc} 
\usepackage[numbers,compress]{natbib}
\usepackage[cmex10]{amsmath}
\usepackage{amsfonts,mathrsfs,mathdots,mathtools,amsthm,amssymb,centernot}
\usepackage[bb=ams]{mathalfa}
\usepackage{algorithmic}
\usepackage{graphicx}
\usepackage{textcomp}
\usepackage[dvipsnames]{xcolor}
\usepackage{tikz}
\usepackage{pgfplots}
\usepackage{subcaption}
\usepackage{dsfont}
\usepackage{pifont}
\usepackage{siunitx}
\usepackage{comment}
\usepackage{pgfplotstable}
\usepackage[ruled]{algorithm2e}
\usepackage{color,soul}
\usepackage{bm}	
\usepackage[inline]{enumitem}
\setlist[enumerate,1]{label = \arabic*.,ref = \arabic*}
\usepackage{nicefrac}
\usepackage{subcaption}
\usepackage{colortbl}
\usepackage{booktabs}
\usepackage{diagbox}
\usepackage{url}
\usepackage{ifthen}
\usepackage{balance}
\usepackage{xfrac}
\usepackage{dirtytalk}
\usepackage{bbm}

\usepackage{hyperref}  
\hypersetup{colorlinks,
            linkcolor=blue,
            citecolor=blue,
            urlcolor=blue,
            anchorcolor=blue,
            filecolor=blue,
            linktocpage,
            plainpages=false,
            breaklinks=true}

\def\supp{\textnormal{supp}}
\def\interior{\textnormal{int}}
\def\epsilon{\varepsilon}

\theoremstyle{definition}
\newtheorem{definition}{Definition}

\newtheorem{theorem}{Theorem}

\newtheorem{corollary}{Corollary}
\newtheorem{proposition}{Proposition}
\newtheorem{example}{Example}
\newtheorem{remark}{Remark}

\mathtoolsset{showonlyrefs}
\allowdisplaybreaks
\IEEEoverridecommandlockouts
\title{Privacy Mechanism Design based on\\Empirical Distributions} 
\begin{document}

\author{%
   \IEEEauthorblockN{Leonhard Grosse\IEEEauthorrefmark{1},
                     Sara Saeidian\IEEEauthorrefmark{1}\IEEEauthorrefmark{2},
                     Tobias J. Oechtering\IEEEauthorrefmark{1},
                      and Mikael Skoglund\IEEEauthorrefmark{1}%
   \IEEEauthorblockA{\IEEEauthorrefmark{1}%
                    KTH Royal Institute of Technology, Stockholm, Sweden,
                     \{lgrosse, saeidian, oech, skoglund\}@kth.se}
       \IEEEauthorblockA{\IEEEauthorrefmark{2}%
                    Inria Saclay, Palaiseau, France}}
\thanks{This work was supported by the Swedish Research Council (VR) under grants 2023-04787 and 2024-06615.}        
} 

\maketitle

\begin{abstract}
Pointwise maximal leakage (PML) is a per-outcome privacy measure based on threat models from quantitative information flow. Privacy guarantees with PML rely on knowledge about the distribution that generated the private data. In this work, we propose a framework for PML privacy assessment and mechanism design with empirical estimates of this data-generating distribution. By extending the PML framework to consider sets of data-generating  distributions, we arrive at bounds on the worst-case leakage within a given set. We use these bounds alongside large-deviation bounds from the literature to provide a method for obtaining distribution-independent $(\varepsilon,\delta)$-PML guarantees when the data-generating distribution is estimated from available data samples. We provide an optimal binary mechanism, and show that mechanism design with this type of uncertainty about the data-generating distribution reduces to a linearly constrained convex program. Further, we show that optimal mechanisms designed for a distribution estimate can be used. Finally, we apply these tools to leakage assessment of the Laplace mechanism and the Gaussian mechanism for binary private data, and numerically show that the presented approach to mechanism design can yield significant utility increase compared to local differential privacy, while retaining similar privacy guarantees.
\end{abstract}


\section{Introduction}
\begin{figure*}
    \centering
    \includegraphics[scale=0.7]{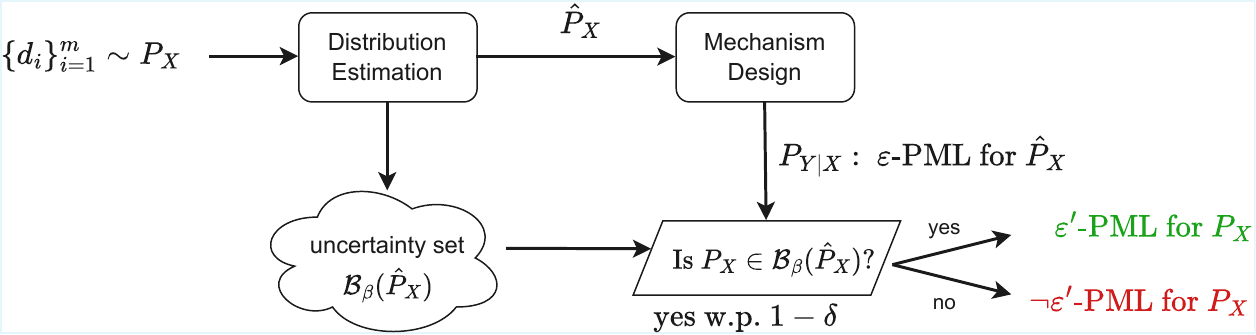}
    \caption{Procedure of obtaining $(\varepsilon,\delta)$-PML guarantees from distribution estimates. Estimation of $P_X$ from $m$ samples yields estimate $\hat P_X$ and an uncertainty set $\mathcal B_\beta(\hat P_X)$ around this estimate in which $P_X$ will lie with probability $1-\delta$. Then, a mechanism is designed to satisfy $\varepsilon$-PML w.r.t. the distribution estimate $\hat P_X$. We bound the PML of that mechanism in the set $\mathcal B_\beta(\hat P_X)$, and obtain $\varepsilon'>\varepsilon$ such that an $\varepsilon'$-PML guarantee will hold if the true distribution $P_X$ is in $\mathcal B_\beta(\hat P_X)$, that is, w.p. $1-\delta$.}
    \label{fig:flowchart}
\end{figure*}

The release of large amounts of data along with increasing computing power enables malicious actors to infer individuals' sensitive features from publicly available information \cite{deanonymNetflix,Alvim2022FlexibleAS}. System designers therefore seek to both measure this vulnerability and to implement effective disclosure control algorithms. 

Let $X$ be a random variable representing some private data, and let 
$P_{Y|X}$ be the \emph{privacy mechanism} which induces the random variable $Y$, representing the privatized version of $X$. 
In this model, \emph{local  differential privacy} (LDP) \cite{duchi2013LDPminmaxDEF} has become the standard tool for assessing privacy, adopted by Apple \cite{appleDP}, Google \cite{googleRAPPOR}, and many others due to its streamlined definitions and easy deployability.  
In the domain of computer security, the framework of \emph{quantitative information flow} (QIF) \cite{smith2009foundations,alvim2020science} was developed as a tool for the assessment of the disclosure risk of sensitive systems. The measures in the QIF framework are built upon adversarial threat models that encode various assumptions about the potential attacker. One such notion is \emph{min-entropy leakage}, which quantifies the guessing advantage of a computationally unbounded adversary that forms a guess $\hat X$ of $X$ after observing the output of a system. By maximizing this quantity over all possible data distributions of $X$, \citet{braun2009quantitative} derive a robust privacy measure later called \emph{maximal leakage}. In \cite{alvim2012measuring} it is shown that this measure can also be derived from a much more general adversarial model, in which an adversary aims to maximize a broad class of gain functions $g(X,\hat X)$. The same measure was introduced to the information theory community by \citet{IssaMaxL}. In their formulation, however, the threat model is adapted to consider an adversary that aims to guess any randomized function of the secret $X$. 

Common among all of the measures in the framework of QIF is the fact that leakage is measured \emph{on-average}. That is, the information leakage of a system is quantified by a single aggregate number. This can significantly weaken the privacy guarantees made with these measures, as e.g. shown in \cite[Section 1]{saeidian2023pointwise}: For outcomes that have low probability, even very large leakage values at this specific outcome might result in a low leakage value on average. This motivates the definition of \emph{pointwise maximal leakage} (PML) in \cite{saeidian2023pointwise}.

As the name suggests, PML is a pointwise adaptation of the threat models of maximal leakage.\footnote{A pointwise view of leakage in QIF referred to as \emph{dynamic min-entropy leakage} has previously been attempted in \cite{espinoza2013min}. However, the authors discard their attempt due to issues with its fundamental properties.} This pointwise formulation enables the treatment of information leakage as a random variable, a view that is also heavily adopted in the differential privacy literature. Further, \citet{saeidian2023inferential} show that these guarantees not only quantify the leakage with respect to the considered threat models, but also offer a clear interpretation of the privacy guarantee. More specifically, it is shown in \cite{saeidian2023inferential} that if a privacy mechanism satisfies $\varepsilon$-PML, the mechanism ensures that all features with min-entropy larger than $\varepsilon$ (that is, more informative features) will not be disclosed to any adversary, regardless of their subjective information. 

PML is a \emph{context-aware} privacy measure. That is, the PML of a mechanism does not only depend on the mechanism itself, but also on the distribution that generated the private data. It is further known that guaranteeing PML for \emph{all} possible data-generating distributions is equivalent to guaranteeing LDP \cite{IssaMaxL}, \cite{10664297}, where the worst-case privacy loss is attained at a degenerated data distribution (one symbol has probability one). Works like \cite{saeidian2023pointwise,grosse2023extremal} show that with the additional assumptions about $P_X$ in PML, the class of mechanisms satisfying a PML constraint is often significantly larger than the class of mechanisms satisfying any (finite) LDP constraint. Specifically, discrete mechanisms for PML can contain zero elements, a property that LDP does not allow for. As a result, mechanisms for PML are able to achieve higher utility compared to similar mechanisms for LDP, often by a large margin \cite{grosse2023extremal,saeidian2023inferential}.

It is argued in, e.g., \cite{10.1145/3433638,10.5555/3454287.3455674}, that the significant decrease in data utility associated with differentially private mechanisms makes LDP impractical for, e.g., some large-scale learning applications. With this issue in mind, the possibility of increased utility by using PML offers a promising alternative for obtaining rigorous privacy guarantees in some practical settings. For practical adoption, however, the dependence of PML on the data-generating distribution requires close attention. Both privacy assessment and mechanism design in previous works incorporated strong assumptions about the data-generating distribution $P_X$, often assuming $P_X$ to be perfectly known. If only partial information about $P_X$ exists, we may find a tradeoff somewhere in between assuming no knowledge about $P_X$ (LDP) and perfect knowledge of $P_X$ (previous works with PML): As outlined above, assuming \emph{no} knowledge about the data distribution can lead to highly conservative privacy assessments. After all, the implicitly considered worst-case distribution might often be very different from the true data-generating distribution. On the other hand, perfect knowledge of $P_X$ is often unavailable in practical applications.

In the context of information-theoretic compression, classical universal source coding approaches \cite[Chapter 13]{Polyanskiy_Wu_2025} face a similar issue, and successfully solve it by \emph{distribution estimation}. That is, in lack of knowledge about the distribution that generated a set of data points, that distribution is instead estimated from the data points, and compression schemes are built based on that distribution estimate. These approaches made entropy coding (which previously required exact knowledge of the data-generating distribution) directly applicable in practice, and have since been widely adopted in contemporary compression schemes (e.g., arithmetic coding in JPEG \cite{rabbani2002overview}). 

In a privacy context, \citet{diaz2019robustness} present results based on distribution estimation and large deviation bounds for multiple on-average privacy measures, including maximal $\alpha$-leakage \cite{liao2019tunable} and min-entropy leakage. Further, \citet{lopuhaa2024mechanisms} design privacy mechanisms for LDP restricted to a known set of distributions. In this work, we build on the work in \cite{grosse2023extremal}, where optimal privacy mechanisms for PML are designed assuming that the system designer knows the data-generating distribution exactly. In many practical scenarios, this is not the case, but the designer may estimate this distribution from observations. Following the approach typical for information-theoretic schemes, we aim to develop methods that allow for the usage of mechanisms designed for PML when only such an empirical estimate of the data-generating distribution is available to the system designer. If the private data is discrete, the presented results offer a way of integrating (partial) knowledge of the data-generating distribution into privacy mechanism design, and can therefore be used to enhance the utility of systems previously designed to satisfy local differential privacy (which assumes no distributional knowledge). The presented procedures are easy to implement, and we numerically show how mechanisms designed in our framework can directly be used in place of mechanisms for LDP, while significantly increasing utility. Hence, the techniques can be implemented to reduce the required noise variance and therefore improve data utility in settings in which some data from the population might be publicly available (e.g., via public medical datasets) or when a portion of a large dataset can be used for estimation. 

\subsection{Overview of Contributions and Outline}
In this section, we provide an informal summary of the key contributions in this work. Detailed derivations are presented in Sections~\ref{sec:approxPML},~\ref{sec:PMLforsets}, and~\ref{sec:optimalbinary}. In Section~\ref{sec:examples}, we further apply the presented tools by analyzing the Laplace and the Gaussian mechanism \cite{dwork2014algorithmic} for binary data. More discussion on related works is presented in Section~\ref{sec:relatedwork}. Detailed proofs of the theorems are presented in the Appendices. Section~\ref{sec:conclusion} concludes the paper. Our contributions are summarized as follows:
\begin{itemize}
    \item We extend the PML framework (previously defined for a \emph{single} data-generating distribution) to instead define privacy for \emph{sets} of distributions (Section \ref{sec:approxPML}).

    \item We analyze the behavior of PML privacy guarantees for these distribution sets, and thereby provide technical tools for leakage assessment when the data-generating distribution changes or is uncertain (Section \ref{sec:PMLforsets}).

    \item We provide a method for mechanism design with PML when the data-generating distribution is not known, but an empirical estimate is available (Section \ref{sec:mainmethod}).

    \item We present closed-form Pareto-optimal mechanisms for binary data, and a linearly constrained convex program for obtaining optimal mechanisms for general discrete distribution estimates (Section \ref{sec:optimalbinary}).

    \item We provide a detailed analysis of the binary Laplace and Gaussian mechanisms in the presented framework, and numerically demonstrate how mechanism design in our framework is able to increase data utility compared to LDP (Section \ref{sec:examples}).
\end{itemize}
The principled approach to mechanism design presented in this paper is illustrated in Figure \ref{fig:flowchart}. We remark that the privacy guarantees in our framework hold up to a probability of failure $\delta$.  However, we will later demonstrate that \emph{(i)} this $\delta$ vanishes with the increase of available samples for distribution estimation and \emph{(ii)} even under general assumptions, the rate at which $\delta$ vanishes is exponential in this sample size. We therefore argue that in settings where large amounts of data are available, our proposed approach can be directly applied to enhance the performance of systems that have previously been hindered by the limitations imposed by strict local differential privacy guarantees, while providing a similar privacy guarantee.

\section{Preliminaries and Notation}
\label{sec:prelim}
We use uppercase letters to denote random variables, lowercase letters to denote their realizations and caligraphic letter to denote (support) sets. Specifically, we use $X\in \mathcal X$ to denote a random variable representing some sensitive data that is to be privatized. If not specified otherwise, we consider discrete alphabets, and we say $P_X$ is the probability mass function of $X$. Further, we use $\mathcal P_{\mathcal X}$ to denote the set of all possible probability distribution on $\mathcal X$. We generally assume that the support of a distribution is known and fixed, and denote it by $\supp(P_X) = \mathcal X = \{1,\dots,N\}$. The set of all distributions with full support on $\mathcal X$ is the interior $\mathcal P_{\mathcal X}$, that is, $\text{int}(\mathcal P_{\mathcal X})$.

Given an outcome $x \in \mathcal X$ of the private random variable $X$, a \emph{privacy mechanism} $P_{Y|X}$ induces the privatized random variable $Y$ by the conditional probability distribution $P_{Y|X=x}$. Let $P_{XY}$ denote the joint distribution of $X$ and $Y$. Then we use $P_{XY} = P_{Y|X} \times P_X$ to imply that $P_{XY}(x,y) = P_{Y|X=x}(y)P_X(x)$ and $P_{Y} = P_{Y|X} \circ P_Y$ to denote the marginalization $P_Y(y) = \sum_{x \in \mathcal X} P_{Y|X=x}(y)P_X(x)$. 

For $N\in\mathbb N$, we let $[N] \coloneqq \{1,..,N\}$ be the set of all positive integers up to $N$, and $\log(\cdot)$ denotes the natural logarithm.
\subsection{Pointwise Maximal Leakage}
The pointwise maximal leakage (PML) privacy measure was recently proposed by \citet{saeidian2023pointwise} as a pointwise generalization of maximal leakage \cite{IssaMaxL,alvim2012measuring}. A key observation that motivates its definition is the fact that the on-average privacy formulation of maximal leakage can critically underestimate the disclosure risk in settings in which there exist outcomes of a privacy mechanism with low probability \cite[Section II]{saeidian2023pointwise}. The pointwise formulation of leakage to each outcome $Y=y$ in PML mitigates this shortcoming. It is shown in \cite{saeidian2023pointwise} that PML arises from two equivalent threat-models, namely the pointwise adaption of the threat model in \cite{alvim2012measuring}, where an adversary aims to maximize an arbitrary non-negative gain function, and the threat model in \cite{IssaMaxL}, in which the adversary aims to maximize the probability of correctly guessing some randomized function of the private data. We restate the more general gain function framework here.
Recall that $X$ is a random variable representing sensitive private data. 

\begin{definition}[Gain function view of PML {{\cite[Corr. 1]{saeidian2023pointwise}}}]
\label{def:PMLgainfunc}
    Let $X$ be a random variable defined on the finite set $\mathcal X$ distributed according to $P_X$. Let $Y$ be the random variable induced by the mechanism $P_{Y|X}$. Then the \emph{pointwise maximal leakage in the gain function view from $X$ to an outcome $y$} is defined as 
    \begin{equation}
        \ell_{P_{XY}}(X\to y) \coloneqq \log \, \sup_g \frac{\sup_{P_{W|Y}}\mathbb E[g(X,W)\mid Y=y]}{\max_{w\in\mathcal W}\mathbb E[g(X,w)]},
    \end{equation}
    where $\mathcal W$ is a finite set, and $g: \mathcal X \times \mathcal W \to \mathbb R_+$ is an arbitrary non-negative gain function.
\end{definition}
This definition can be interpreted as follows: Assume an adversary aims to guess the value of the secret $X$, and hence forms a guess $w\in \mathcal W$, where $\mathcal W$ represents some guessing space (e.g., one can take $\mathcal W = \mathcal X$ for illustration). The adversary then quantifies the quality of her guess by some non-negative gain function $g$. Through its arbitrariness, this gain function encompasses a wide array of attacks, including, e.g., membership inference attacks and reconstruction attacks (see \cite{saeidian2023pointwise} for examples of the corresponding gain functions). To quantify the information leaking through $P_{Y|X}$, Definition~\ref{def:PMLgainfunc} then compares the maximum expected \emph{posterior} gain $\sup_{P_{W|Y}}\mathbb E[g(X,W)|Y=y]$ of a guess after observing $Y=y$ with the maximum expected \emph{prior} gain $\max_{w \in \mathcal W}\mathbb E[g(X,w)]$ of a guess made without observing the outcome of the privacy mechanism. To obtain robustness with respect to the specific gain function used by the adversary, this ratio is then maximized over all functions $g$ mapping from $\mathcal X \times \mathcal W$ to the bounded non-negative real numbers. It can then be shown that this threat model definition admits a simple representation in terms of the Rényi-divergence of order infinity.
\begin{theorem}[{{\cite[Thm. 1]{saeidian2023pointwise}}}]
    Let $X,Y$ be random variables distributed according to the joint distribution $P_{XY}$. Let $P_X \in \interior (\mathcal P_{\mathcal X})$ be the marginal distribution of $X$. The pointwise maximal leakage according to Definition~\ref{def:PMLgainfunc} simplifies to
    \begin{equation}
        \ell_{P_{XY}}(X\to y) = D_\infty(P_{X|Y=y}||P_X),
    \end{equation}
    where $P_{X|Y=y}$ denotes the conditional distribution of $X$ given an outcome $Y=y$, and $D_\infty$ denotes the Rényi-divergence of order infinity \cite{renyi1961entropy}, that is, 
    \begin{align}
        D_\infty(P_{X|Y=y}||P_X) &= \log\,\max_{x\in\mathcal X}\,\frac{P_{X|Y=y}(x)}{P_X(x)} \\&= \log\,\max_{x\in\mathcal X}\,\frac{P_{Y|X=x}(y)}{P_Y(y)} ,
    \end{align}
    where $P_Y = P_{Y|X} \circ P_X$.
    \end{theorem}
    When the mechanism $P_{Y|X}$ and the data distribution $P_X$ are clear from context, we occasionally write $\ell(X \to y)$ to mean $\ell_{P_{Y|X}\times P_X}(X\to y)$. Notably, PML satisfies the common desirable properties of privacy measures like pre-processing, post-processing and linear composition (see \cite{saeidian2023pointwise} for details). 
    \subsubsection{Privacy guarantees}
    Since PML considers privacy leakage to each individual outcome $y$ of the privacy mechanism, the leakage with respect to the random variable $Y$ becomes a random variable itself. Different privacy guarantees can be formulated by restricting the statistics of this random variable. The most stringent definition, which we will refer to as $\varepsilon$-PML, is defined as follows.
    \begin{definition}[Almost-sure guarantee, {{\cite[Def. 4]{saeidian2023pointwise}}}]
    \label{def:almostsurely}
        For any arbitrary but fixed data distribution $P_X$ and some $\varepsilon \geq 0$, a mechanism is said to satisfy \emph{$\varepsilon$-PML with respect to $P_X$} if 
        \begin{equation}
            \mathbb P_{Y\sim P_Y}\{\ell(X\to Y) \leq \varepsilon\} = 1.
        \end{equation}
    \end{definition}
    A slightly weaker probabilitstic guarantee can be defined by allowing the $\varepsilon$-bound to fail with probability $\delta$.
    \begin{definition}[Tail-bound guarantee {{\cite[Def. 5]{saeidian2023pointwise}}}]
    \label{def:tailbound}
        Given some arbitrary but fixed data distribution $P_X$, a mechanism $P_{Y|X}$ is said to satisfy $(\varepsilon,\delta)$-PML if 
        \begin{equation}
            \mathbb P_{Y\sim P_Y}\{\ell(X\to Y) \leq \varepsilon\} \geq 1-\delta.
        \end{equation}
    \end{definition}
    Note that according to Definitions~\ref{def:almostsurely} and~\ref{def:tailbound}, $(\varepsilon,0)$-PML is equivalent to $\varepsilon$-PML. Further, all mechanisms satisfy $(\varepsilon,1)$-PML, and (for discrete alphabets) $\varepsilon_{\max}$-PML, where $\varepsilon_{\max} = -\log(\min_xP_X(x))$ \cite{saeidian2023pointwise}.

In order to characterize the set of privacy mechanisms satisfying a certain $\varepsilon$-PML guarantee according to Definition~\ref{def:almostsurely}, we define
\begin{align}
    \varepsilon_{\min}&(P_{Y|X},P_X) \\&\coloneqq \inf\{\varepsilon \geq 0: \mathbb P[\ell_{P_X\times P_{Y|X}}(X\to Y) \leq \varepsilon] = 1\}
\end{align}
to be the minimum value of $\varepsilon$ such that the mechanism $P_{Y|X}$ satisfies $\epsilon$-PML with respect to the data distribution $P_X$. In the case of discrete alphabets, the privacy mechanisms can be written as row-stochastic matrices representing the transitioning kernels. For $N \in \mathbb N$, $M\leq N$, let $\mathcal S_{N\times M} \in [0,1]^{N\times M}$ be the set of all $N \times M$ row-stochastic matrices. We then write
\begin{equation}
    \mathcal M(\varepsilon,P_X) \coloneqq \biggl{\{}P_{Y|X}\in \bigcup_{M=1}^N \mathcal S_{N\times M}: \varepsilon_{\min}(P_{Y|X},P_X) \leq \varepsilon \biggr{\}}
\end{equation}
for the set of all mechanisms that satisfy $\varepsilon$-PML with respect to the data distribution $P_X$.\footnote{It is shown in \cite[Theorem 1]{grosse2023extremal} that the restriction $M\leq N$ does not have any effect on the achievable utility of the mechanisms in $\mathcal M(\varepsilon,P_X)$.} As shown in \cite[Lemma 1]{grosse2023extremal}, for any $\varepsilon \geq 0$, this set is a closed and bounded convex polytope. We extend this definition to \emph{sets} of data distributions $\mathcal P \subseteq \text{int}(\mathcal P_{\mathcal X})$ by requiring that mechanisms in $\mathcal M(\varepsilon,\mathcal P)$ satisfy $\varepsilon$-PML with respect to \emph{all} distributions in the set $\mathcal P$, that is,
\begin{equation}
\label{eq:MmathcalP}
    \mathcal M(\varepsilon,\mathcal P) = \bigcap_{P_X \in \mathcal P} \mathcal M(\varepsilon,P_X).
\end{equation}

\subsubsection{Extremal Mechanism}
\label{subsec:extremalmech}
The results in \cite{grosse2023extremal} show that so called \emph{extremal mechanisms} are optimal solutions to the privacy-utility tradeoff between $\varepsilon$-PML and \emph{sub-convex} utility functions for categorical data. These utility functions include all measures of dependence between $X$ and $Y$ represented by a function $U: \mathbb R_+^{|\mathcal X|\times |\mathcal Y|}\times \mathbb R_+^{|\mathcal X|} \to \mathbb R_+$ that can be written as
\begin{equation}
\label{eq:subconvexdef}
    U(P_{Y|X},P_X) = \sum_{y\in\mathcal Y}\mu\big([P_{Y|X=x}(y)]_x,P_X\big),
\end{equation}
where $\mu: \mathbb R_+^{|\mathcal X|} \times \mathbb R^{|\mathcal X|}\to \mathbb R_+$ is a \emph{sub-linear} function.\footnote{A function is said to be sub-linear if it is homogenuous and convex. The definition of sub-convex utility functions was originally presented in \cite{extremalmechanismLong}.} This class of functions notably includes the mutual information $I(X;Y)$ as well as many other common dependence measures. The authors show that the optimal mechanisms in this setup constitute extreme points of the convex polytope $\mathcal M(\varepsilon,P_X)$. This in particular implies that if $P^*_{Y|X}$ is an extremal mechanism satisfying $\varepsilon$-PML with respect to $P_X$, then
\begin{equation}
    \ell_{P^*_{Y|X}\times P_X}(X\to y) = \varepsilon \quad \forall y \in \mathcal Y.
\end{equation}
Closed form extremal mechanisms for specific configurations of $\varepsilon$ and $P_X$ are then presented in \cite[Theorem 2-4]{grosse2023extremal}. Further, \cite[Theorem 6]{grosse2023extremal} presents a linear program with which optimal mechanisms can be computed for arbitrary $\varepsilon$ and $P_X$. 

Another crucial fact observed in \cite{grosse2023extremal} is that the value of the privacy parameter determines the maximum number of zero-valued entries in the mechanism matrix. This leads to the definition of \emph{PML privacy regions}, which partition the interval $[0,\varepsilon_{\max})$ into non-overlapping subspaces.
\begin{definition}[PML privacy regions {{\cite{grosse2023extremal}}}]
\label{def:privacyregions}
    Assume that $P_X(x_1) \geq P_X(x_2) \geq \dots \geq P_X(x_N)$. For $k \in [N-1]$, define $\varepsilon_k (P_X) = - \log \sum_{i=1}^{N-k}P_X(x_i)$. An $\varepsilon$-PML guarantee is said to be in the $k^{\text{th}}$ privacy region if $\varepsilon \in [\varepsilon_{k-1}(P_X),\varepsilon_k(P_X))$.
\end{definition}
Using this definition, \cite[Lemma 2]{grosse2023extremal} then states that any mechanism $P_{Y|X} \in \mathcal M(\varepsilon,P_X)$, with $\varepsilon$ in the $k^{\text{th}}$ privacy region, can have at most $k-1$ zero-valued entries in each column. An important special case of this lemma is the first privacy region, in which a mechanism matrix can only have strictly positive entries.

\subsection{Distribution Estimation}
\label{sec:distest}
PML poses a robust framework in terms of adversarial assumptions. However, assuming perfect knowledge about the private data's distribution can often be highly impractical. Definition~\ref{def:PMLgainfunc} measures information leakage only under the assumption of this specific distribution. In the context of modern data processing systems, arguably the most common scenario involves data distributions that are not perfectly known, but can only be estimated from the available samples in a data set. To this end, we utilize the large deviation bounds for distribution estimation as presented by \citet{weissman2003inequalities}, which we restate below.

Consider a dataset $\{d_i\}_{i=1}^m$ consisting of $m$ entries $d_i$ independently sampled from the data source $X$. The goal of the system designer is to give privacy guarantees for the underlying data distribution $P_X$, while only having access to the \emph{empirical estimate} of the distribution $\hat P_X$, which she computes from $\{d_i\}_{i=1}^m$ as 
\begin{equation}
\label{eq:empiricaldistribution}
    \hat P_{X}(x) = \frac{1}{m}\sum_{i=1}^m \mathbbm 1 \{d_i =x\} \quad \forall x \in \mathcal X.
\end{equation}
Define the $\ell_1$-distance between discrete distributions as
\begin{equation}
    ||P_X - \hat{P}_{X}||_1 = \sum_{i=1}^N |P_X(x_i) - \hat P_X(x_i)|.
\end{equation}
Further, define the function $\varphi: \mathbb [0,0.5) \to (2,\infty)$ as 
\begin{equation}
    \varphi(p) \coloneqq \frac{1}{1-2p}\log\frac{1-p}{p}.
\end{equation}
Note that $\varphi$ is decreasing on $(0,0.5)$ and by continuity we define $\varphi(0.5) \coloneqq 2$.
Further, define $\kappa_{P_X}: \mathcal P_{\mathcal X} \to (0,0.5]$ as
\begin{equation}
    \kappa_{P_X} = \max_{A \subseteq \mathcal X}\min \{P_X(A), 1-P_X(A)\}.
\end{equation}
The following result, presented in \cite{weissman2003inequalities}, states that the probability of the true distribution not being close in $\ell_1$-distance to the estimate decays exponentially with the number of samples.
\begin{theorem}[see {{\cite[Thm. 1]{weissman2003inequalities}}}]
\label{thm:weissman}
    Let $X_i \sim P_X$ be i.i.d. random variables defined on $\mathcal X$ with $|\mathcal X| = N$. Let $\hat P_X$ be estimated from $m$ realizations $X_1,\dots,X_m$ by \eqref{eq:empiricaldistribution}. Then we have
\begin{equation}
    \mathbb P\{||P_X - \hat P_X||_1 \geq \beta\} \leq (2^N-2)\exp\biggl(-m\frac{\varphi(\kappa_{P_X})\beta^2}{4}\biggr).
\end{equation}
\end{theorem}

\section{Approximate Pointwise Maximal Leakage}
\label{sec:approxPML}
The bound in Theorem~\ref{thm:weissman} can be used to obtain $(\varepsilon,\delta)$-PML guarantees when only an empirical estimate of the true data distribution is available to the system designer. The process is as follows. Given the distribution estimate $\hat P_X$, design a mechanism that satisfies $\varepsilon$-PML with respect to \emph{any} data distribution in the $\ell_1$-ball $\mathcal B_\beta(\hat P_X) = \{P_X: ||\hat P_X-P_X||_1\leq \beta\}$. That is, design a mechanism such that 
\begin{equation}
\label{eq:PMLinsidel1ball}
    \varepsilon_{\min}(P_{Y|X},P_X) \leq \varepsilon \quad \forall P_X \in \mathcal B_\beta(\hat P_X).
\end{equation}
Given that the true data distribution is inside of $\mathcal B_\beta(\hat P_X)$, the mechanism will therefore satisfy $\varepsilon$-PML. To obtain a privacy guarantee that is independent of the true data distribution, what remains is to bound the probability that the estimated distribution is \say{bad} in the sense that the true distribution does \emph{not} lie inside of $\mathcal B_\beta(\hat P_X)$. Define this probability as 
\begin{equation}
    \delta = \mathbb P\big[||\hat P_X-P_X||_1 > \beta\big].
\end{equation}
Then a mechanism designed such that \eqref{eq:PMLinsidel1ball} holds will \emph{not} satisfy $\varepsilon$-PML with probability $\delta$. Hence, under any circumstances (\say{bad} estimate or not) the mechanism satisfies $(\varepsilon,\delta)$-PML.

More practically, for some desired maximum acceptable probability of failure $\delta$, one can obtain the required $\ell_1$-radius $\beta$ of the uncertainty set by noticing that the bound in Theorem~\ref{thm:weissman} implies the lower bound
\begin{equation}
    \beta \geq 2\sqrt{\frac{\log(2^N-2)-\log\delta}{m \varphi(\kappa_{P_X})}}.
\end{equation}
Note that this bound still requires knowledge of the underlying distribution due to its dependence on $\varphi(\kappa_{P_X})$. However, a distribution-independent bound can be obtained by observing that $\varphi(\kappa_{P_X}) \geq 2$, which yields the distribution-independent lower bound
\begin{equation}
\label{eq:distindepbeta}
    \beta \geq \sqrt{\frac{2}{m}(\log(2^N-2)-\log \delta)}.
\end{equation}
Bounding $\varphi(\kappa_{P_X})$ in this way tends to be relatively tight for larger alphabets. The bound also holds with equality if there exists a set $A \subset \mathcal X$ such that $P_X(A) = \nicefrac{1}{2}$. This class of probability distributions notably includes the uniform distribution for even alphabet sizes, that is, if $N\, \text{mod}\,2 = 0$.

\section{PML for Sets of Distributions}
\label{sec:PMLforsets}
In order to provide privacy guarantees with the techniques described above, we first study PML as a function of the data distribution. This will enable us to provide bounds on the privacy leakage for mechanisms given a set of possible distributions. We begin by characterizing the set $\mathcal M(\varepsilon, \mathcal P)$ of mechanisms satisfying $\varepsilon$-PML for a set of distributions $\mathcal P$. Then, we present a local Lipschitz bound on the leakage increase for general discrete mechanisms. We further present a more specialized bound that aims to address the above scenario, in which an estimate of the data distribution is available, and privacy mechanisms are designed for this estimate. 
\subsection{Behavior of PML Under Uncertain Data Distributions}
\label{sec:conv}
We begin by characterizing the behavior of $\varepsilon$-PML guarantees with respect to the data-generating distribution $P_X$.
\begin{proposition}
\label{prop:epsPMLisconvex}
    Both $\ell_{P_{Y|X}\times P_X}(X \to y)$ and $\epsilon_{\min}(P_{Y|X},P_X)$ are convex in $P_X \in \text{int}(\mathcal P_{\mathcal X})$ for any fixed mechanism $P_{Y|X}$.
\end{proposition}
\begin{proof}
    We have $\ell_{P_{Y|X}\times P_X}(X \to y) = \max_{x\in\mathcal X}i(x;y)$ and $\epsilon_{\min}(P_{Y|X},P_X) = \max_{y\in \supp(P_Y)}\ell_{P_X \times P_{Y|X}}(X \to y)$, where $i(x;y)$ denotes the \textit{information density} \cite{Polyanskiy_Wu_2025}. Since $i(x;y)$ is convex in $P_X$ for fixed $P_{Y|X}$ \cite{Polyanskiy_Wu_2025}. The statement follows since pointwise maxima of convex functions are convex \cite{boyd2004convex}.
\end{proof}

Proposition \ref{prop:epsPMLisconvex} implies that the worst-case leakage values are attained at the boundary of a set of data distributions. If that set is a convex polytope, we have the following corollary.

\begin{corollary}
\label{corr:extremalSet}
    Let $\mathcal P$ be a convex set of data distributions with finitely many extreme points $\mathcal{P}^{*} \coloneqq \{\pi^{(1)},\dots,\pi^{(k)}\}$. Then $\mathcal M(\epsilon,\mathcal{P}) = \mathcal M(\epsilon,\mathcal{P}^{*})$.
\end{corollary}

This result has an important consequence for mechanism design with uncertain distributions: When searching for a mechanism that maximizes some convex utility function (in particular any sub-convex utility function), it is enough to consider mechanisms satisfying $\varepsilon$-PML with respect to the extreme points of the set $\mathcal P$, instead of the entire (often infinite) set.  This result can be used to find closed-form solutions for optimal mechanisms, as presented in Section~\ref{sec:optimalbinary}.

Next, we define the notion of \emph{$\mathcal P$-local leakage capacity}, which quantifies the worst-case leakage (in terms of $\epsilon$-PML) of the mechanism $P_{Y|X}$ for any data distribution in a set $\mathcal P$.

\begin{definition}[$\mathcal P$-local leakage capacity]
\label{def:localcap}
    Given a set of data distributions $\mathcal P \subseteq \text{int}(\mathcal P_{\mathcal X})$, we define a mechanism's \emph{$\mathcal P$-local leakage capacity} $C(P_{Y|X},\mathcal P)$ as
    \begin{align}
        C(P_{Y|X},\mathcal P) &\coloneqq \sup_{Q_X \in \mathcal P} \varepsilon_{\min}(P_{Y|X},Q_X) \\&= \max_y \sup_{Q_X \in \mathcal P} \log \frac{\max_{x} P_{Y|X=x}(y)}{\sum_{x'} Q_X(x')P_{Y|X=x'}(y)}.
    \end{align}
\end{definition}

\begin{remark}
    If $\mathcal P = \text{int}(\mathcal P_{\mathcal X})$, that is, if the uncertainty set contains all distributions with full support on $\mathcal X$, the $\mathcal P$-local leakage capacity of a mechanism is equivalent to the leakage capacity as defined by \citet{saeidian2023inferential}, which also coincides with the logarithm of the \emph{lift-capacity} in \cite{10664297}.  
\end{remark}

With this definition of leakage capacity at hand, we can prove a local Lipschitz bound on the $\epsilon$-PML guarantees of a privacy mechanism with respect to different data distributions.

\begin{theorem}
\label{thrm:lipschitz}
    Consider the two metric spaces $(\mathcal P_{\mathcal X},||\cdot||_1)$ and $(\mathbb R_+, |\cdot|)$. Fix some subset $\mathcal P \subseteq \text{int}(\mathcal P_{\mathcal X})$ and a privacy mechanism $P_{Y|X}$. Then the function $\epsilon_{\min}: \interior(\mathcal P_{\mathcal X}) \to \mathbb R_+$ is $\exp(C(P_{Y|X},\mathcal P))$-locally Lipschitz on the set $\mathcal P$, that is, for any two distributions $P_X,\,Q_X \in \mathcal P$, we have
    \begin{align}
        |\epsilon_{\min}(P_{Y|X},P_X) - &\epsilon_{\min}(P_{Y|X},Q_X)| \\[.5em]&\leq \exp\Big({C(P_{Y|X},\mathcal P)\Big)}||P_X - Q_X||_1.
    \end{align}
\end{theorem}
\begin{proof}
  See Appendix~\ref{app:lipschitzproof}.
\end{proof}
The result in Theorem~\ref{thrm:lipschitz} holds true for any privacy mechanism, and any two distributions in the set $\mathcal P$. It therefore shows that given a finite leakage capacity, the PML of a mechanism changes smoothly over the data-generating distributions $P_X$ in that set. While this result is mainly of theoretical interest (e.g., for proving convergence results \cite{diaz2019robustness} or in privacy auditing \cite{sreekumar2023limit}), in what follows, we will derive a more practically-oriented bound for the setting of mechanism design.

\subsection{Leakage Sensitivity}
 From a design perspective, cases in which a mechanism is designed with respect to a known reference distribution are often of interest. If there is precise knowledge about this design distribution, a sharper bound on the leakage increase can be obtained. To this end, we introduce the \emph{local leakage sensitivity}, which quantifies the sensitivity of a mechanism's leakage to distributional changes. Given that the PML of the mechanism is known for a specific reference distribution (e.g. an empirical estimate), the local leakage sensitivity quantifies how much the leakage might increase when the input distribution deviates from the reference by a specified $\ell_1$-distance.
\begin{definition}[$(P_X,\mathcal P)$-local leakage sensitivity]
\label{def:leakagesensitivity}
    For some distribution $P_X \in \mathcal P \subseteq \text{int}(\mathcal P_{\mathcal X})$, we define the \emph{local leakage sensitivity} $S_{P_X}(P_{Y|X},\mathcal P)$ of a mechanism $P_{Y|X}$ as
    \begin{align}
        S_{P_X}(P_{Y|X},\mathcal P) \coloneqq \sup_{Q_X \in \mathcal P} \Big(\varepsilon_{\min}&(P_{Y|X},Q_X) \\&- \varepsilon_{\min}(P_{Y|X},P_X)\Big).
    \end{align}
\end{definition}

\begin{remark}
\label{rem:relationbetweenCandS}
     The definition of local leakage capacity in Definition~\ref{def:localcap} underlines the operational interpretation for the local leakage sensitivity of a mechanism: Assume that the mechanism $P_{Y|X}$ satisfies $\epsilon$-PML for some distribution $P_X$ in $\mathcal P$. Then the $\mathcal P$-local leakage capacity can be bounded as 
    \begin{align}
        C(P_{Y|X},\mathcal P) &= \max_y \sup_{Q_X \in \mathcal P}\log \frac{\max_x P_{Y|X=x}(y)}{\sum_{x'} Q_X(x') P_{Y|X=x'}(y)} \\[.7em]&\leq \epsilon + \max_y \sup_{Q_X \in \mathcal P}\log \frac{\sum_x P_X(x)P_{Y|X=x}(y)}{\sum_{x'} Q_X(x')P_{Y|X=x'}(y)} \\[.1em]&= \epsilon + S_{P_X}(P_{Y|X},\mathcal P).
    \end{align}
\end{remark}

\subsection{$\ell_1$-Uncertainty sets}
\begin{figure*}[!ht]
    \centering
    \begin{subfigure}[b]{.99\columnwidth}
    \centering
        \includegraphics[scale=0.5]{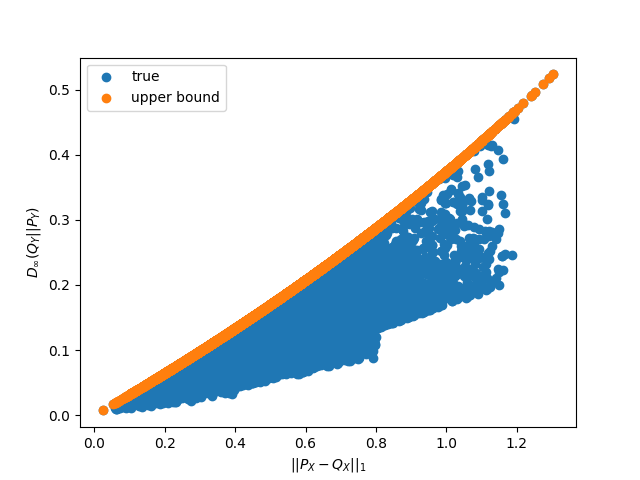}
        \caption{$P_{Y|X}$ in Example~\ref{ex:boundExtrMech}.}
        \label{subfig:TIFS1}
    \end{subfigure}
    \begin{subfigure}[b]{0.99\columnwidth}
    \centering
        \includegraphics[scale=0.5]{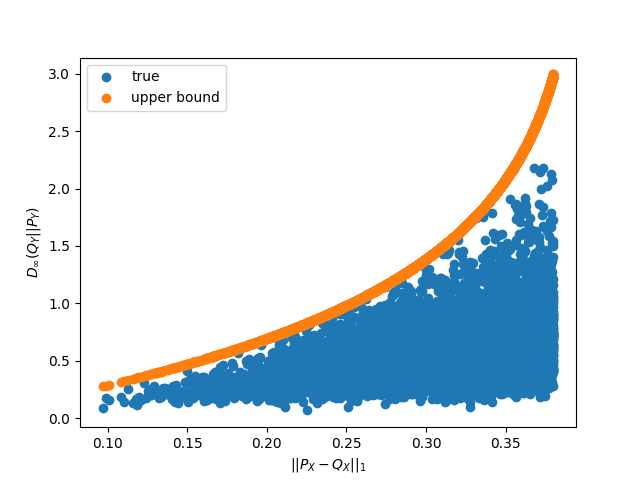}
        \caption{Example~\ref{ex:boundksingular}: $2$-singular mechanism, $N=10$.}
        \label{subfig:2singular}
    \end{subfigure}
    \caption{The upper bound \eqref{eq:Spx_l1bound} given in Theorem~\ref{thm:l1gamma} for Examples~\ref{ex:boundExtrMech} and~\ref{ex:boundksingular}. For the reference distribution $P_X$ for which each mechanism is optimal, $n=10000$ random distributions $\{Q_X^{(i)}\}_{i=1}^n$ are generated. For each $Q_X^{(i)}$, the value of $D_\infty(P_Y||Q_Y^{(i)})=D_\infty(P_{Y|X}\circ P_X||P_{Y|X}\circ Q_X^{(i)})$ (labeled ``true" in the figures) is plotted over the $\ell_1$-distance $\beta_i \coloneqq ||P_X-Q^{(i)}_X||_1$ between the distributions. Note that $S_{P_X}(P_{Y|X},\mathcal B_{\beta_i}(P_X))$ is given by the maximum of $D_\infty(P_Y||P_{Y|X}\circ Q_X)$ for any $Q_X \in \mathcal B_{\beta_i}(P_X)$.}
    \label{fig:upperbound}  
\end{figure*}
In this section, we will derive bounds on the leakage sensitivity for the distribution estimation scenario central to the overall approach presented in this work: Assume an estimate of the data-generating distribution is computed from \eqref{eq:empiricaldistribution}. The uncertainty set $\mathcal P$ then becomes an $\ell_1$-ball with radius $\beta$ around this estimate, that is, we define
    \begin{equation}
    \label{eq:l1uncertaintyset}
        \mathcal P = \mathcal B_\beta(\hat P_X) \coloneqq \{P_X \in \interior(\mathcal P_{\mathcal X}): ||\hat P_X - P_X||_1 \leq \beta\}.
    \end{equation}
The symmetry of this set can be used to obtain a bound on the leakage sensitivity which only depends on the $\epsilon$-PML guarantee of the mechanism with respect to the estimated distribution, as well as the $\ell_1$-radius of the uncertainty set. 

\begin{theorem}
\label{thm:l1gamma}
    Let $P_X \in \interior(\mathcal P_{\mathcal X})$ and define $p_{\min} = \min_x P_X(x)$. For some $\beta < 2p_{\min}$, let $\mathcal B_\beta(P_X)$ be defined by \eqref{eq:l1uncertaintyset}.
    If $P_{Y|X}$ satisfies $\epsilon$-PML with respect to $P_X$, where $\epsilon$ is in the $k^{\text{th}}$ privacy region according to Definition~\ref{def:privacyregions},\footnote{Recall that $\varepsilon$ in \emph{any} privacy region implies $\varepsilon\leq\varepsilon_{\max}$. This enforces a non-negative term in the logarithm in \eqref{eq:Spx_l1bound}. Since all mechanisms (including releasing data unchanged) satisfy $\varepsilon_{\max}$-PML, this is no restriction.} then,
    \begin{equation}
    \label{eq:Spx_l1bound}
          S_{P_X}(P_{Y|X},\mathcal B_\beta(P_X)) \leq \begin{cases}
              -\log \big(1-\frac{\beta}{2}\frac{e^{\epsilon}-1}{p_{\min}}\big), &\text{if }k=1, \\[0.3em]
              -\log \big(1-\frac{\beta e^{\epsilon}}{2}\big), &\text{if }k>1,
          \end{cases}
    \end{equation}
    with equality for $k=1$ if $P_{Y|X}$ is extremal (see Section~\ref{subsec:extremalmech}), and for $k>1$ if for all $y\in \mathcal Y$, there exists $x\in\mathcal X$ such that we have $P_{Y|X=x}(y)=0$.  
\end{theorem}
\begin{proof}
   See Appendix~\ref{app:l1gammaproof}.
\end{proof}
The above result can be used to obtain bounds on $\epsilon_{\min}(P_{Y|X},Q_X)$ in the sets $\mathcal B_\beta(P_X)$. That is, Theorem~\ref{thm:l1gamma} enables us to upper bound the increase in leakage for any $\epsilon$-PML mechanism due to a change in data distribution. 
\begin{corollary}
\label{corr:leakagebound}
    Let $P_{Y|X}$ be a mechanism satisfying $\epsilon$-PML with respect to the distribution $P_X$. For any $Q_X \in \mathcal B_\beta(P_X)$ with $\beta < 2e^{-\varepsilon}$, the PML with respect to $Q_X$ is bounded by\footnote{Note that this bound can be improved in the first privacy region by using the specific bound for this case in Theorem~\ref{thm:l1gamma}. However, since the results for the other privacy regions also hold for privacy region one, we only present this more general case here for sake of clarity.}
    \begin{equation}
        \varepsilon_{\min}(P_{Y|X},Q_X) \leq \epsilon + \log \bigg(1-\frac{e^{\epsilon} ||P_X-Q_X||_1}{2}\bigg)^{-1}.
    \end{equation}
\end{corollary}

The following examples illustrate the bound in Theorem~\ref{thm:l1gamma}. Note that simple manipulations may be used to show that $S_{P_X}(P_{Y|X},\mathcal P) = \max_{Q_X \in \mathcal P}D_\infty(P_{Y|X}\circ P_X||P_{Y|X}\circ Q_X)$. To evaluate the tightness of the upper bound presented in Theorem~\ref{thm:l1gamma}, we therefore compare it to the true values of $D_\infty$ for randomly generated distributions $Q_X$ as a function of the $\ell_1$-distance between the distribution pairs $P_X$, $Q_X$ in Figure~\ref{fig:upperbound}.  
\begin{example}[First privacy region extremal mechanism]
\label{ex:boundExtrMech}
    Let $P_{Y|X}$ be the mechanism in \cite[Example 1]{grosse2023extremal}, that is, let
    \begin{equation}
        P_{Y|X} = \begin{bmatrix}
            0.325 & 0.225 & 0.225 & 0.225 \\
            0.45 & 0.1 & 0.225 & 0.225 \\
            0.45 & 0.225 & 0.1 & 0.225 \\
            0.45 & 0.225 & 0.225 & 0.1
        \end{bmatrix}.
    \end{equation}
    $P_{Y|X}$ is an optimal mechanism for sub-convex functions and $\varepsilon= \log\nicefrac{9}{8}$, $P_X = (\nicefrac{2}{5},\nicefrac{1}{5},\nicefrac{1}{5},\nicefrac{1}{5})$, which constitutes a mechanism in the first privacy region. The bound in Theorem~\ref{thm:l1gamma} for this mechanism is shown in Figure~\ref{subfig:TIFS1}.
\end{example}

\begin{example}[$k$-singular mechanisms]
\label{ex:boundksingular}
    Inspired by \cite{altuug2014singulardistributions}, we say a mechanism is $k$-singular if it is doubly stochastic and each entry in the mechanism matrix is either 0 or $1/k$, that is,
     \begin{equation}
        P_{Y|X=x}(y) = \begin{cases}
            \frac{1}{k}, \text{ if } y \in \supp(P_{Y|X=x}) \\
            0, \text{ otherwise.}
        \end{cases}
    \end{equation}$k$-singular mechanisms are optimal extremal mechanisms for the uniform distribution satisfying $\varepsilon$-PML with $\epsilon = \log(N/k)$ \cite[Theorem 4]{grosse2023extremal}. The resulting bound using Theorem~\ref{thm:l1gamma} for one such mechanism with $N=10$ is shown in Figure~\ref{subfig:2singular}.
\end{example}

\subsection{$(\epsilon,\delta)$-PML Guarantees for Distribution Estimates}
\label{sec:mainmethod}
\begin{figure*}[!t]
\centering
\begin{subfigure}{.49\textwidth}
    \centering
    \includegraphics[scale=0.425]{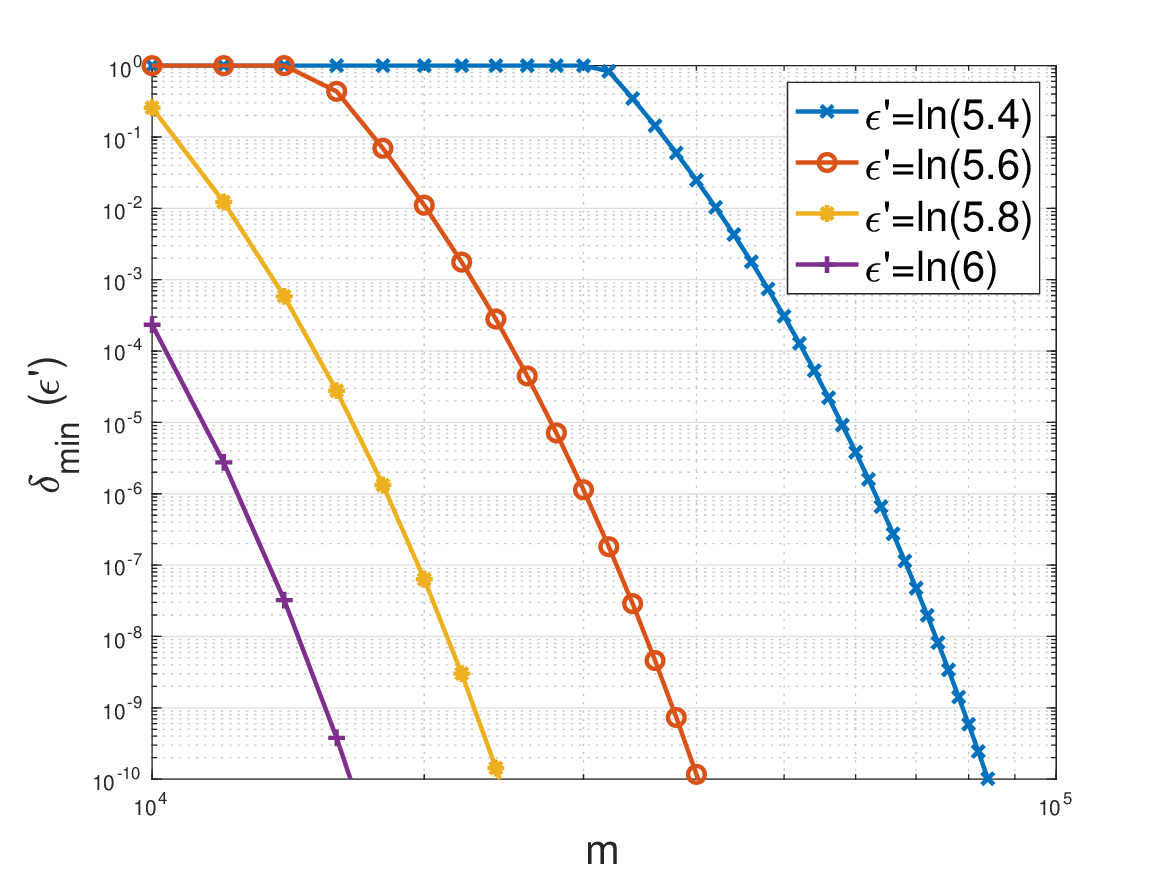}
    \caption{$N=20$, $\epsilon = \log 5$.}
    \label{fig:deltaminplot:subfig:deltaeps}
\end{subfigure}
    \begin{subfigure}{.49\textwidth}
    \centering
    \includegraphics[scale=0.425]{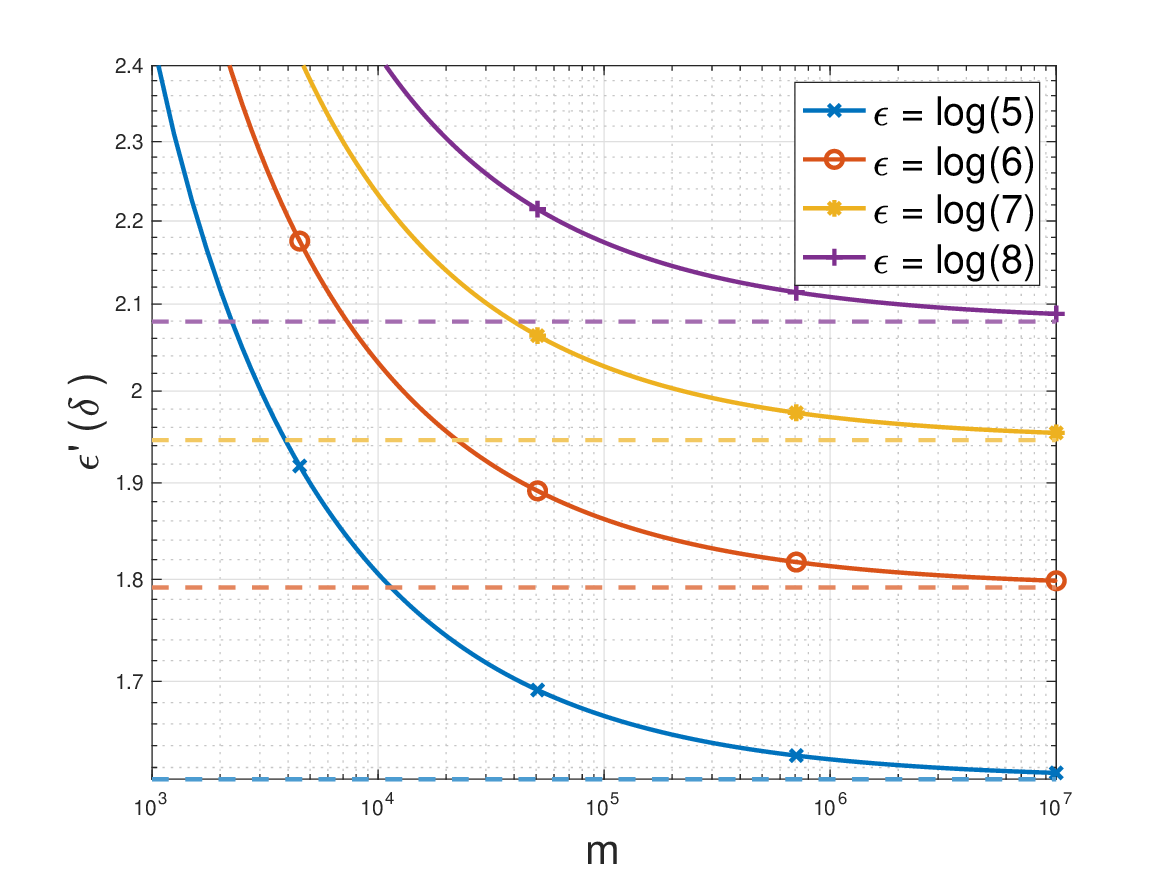}
    \caption{$N=20$, $\delta = 10^{-5}$.}
    \label{fig:deltaminplot:subfig:epsdelta}
\end{subfigure}
    \caption{(\subref{fig:deltaminplot:subfig:deltaeps}): Tradeoff between leakage increase and probability of failure according to Corollary~\ref{corr:deltamin}. Whenever the bound yields values above one, it is vacuous and we assume $\delta_{\min} = 1$. (\subref{fig:deltaminplot:subfig:epsdelta}): Leakage increase according to \eqref{eq:epsmax} for a fixed value of $\delta = 10^{-5}$ and different design parameters $\epsilon$ as a function of the amount of samples $m$ used to estimate the data-generating distribution. Dashed lines mark the the privacy parameter under the estimated distribution.}
    \label{fig:deltaminplot}
\end{figure*}
Using the presented bounds, we can proceed to find the optimal $(\epsilon,\delta)$-PML tradeoff for a fixed mechanism design: Consider a mechanism $P_{Y|X}$ has been designed to satisfy $\epsilon$-PML with respect to a distribution $\hat P_X$ estimated from $m$ data samples according to \eqref{eq:empiricaldistribution}. With the distribution estimation bound in Theorem~\ref{thm:weissman} at hand, we can find the minimum $\delta$ that can be achieved for some fixed leakage target $\varepsilon' \geq \varepsilon$, that is,
\begin{equation}
    \delta_{\min}(\epsilon') \coloneqq \mathbb P\biggl[\bigl{\{} y\in\mathcal Y: \ell_{P_X \times P_{Y|X}}(X \to y) > \epsilon'\bigr{\}}\biggr],
\end{equation}
where $P_X$ is the true distribution that generated the $m$ observed data samples (from which $\hat P_X$ was estimated). Conversely, me may also fix some $\delta \in [0,1]$ and find $\varepsilon'(\delta)$ such that a mechanism satisfies $\varepsilon'(\delta)$-PML w.p. $1-\delta$. Let us begin by examining the latter case. To utilize the statement of Theorem~\ref{thm:weissman}, we need to be able to translate between the probability of failure $\delta$, and the required radius $\beta$ of the uncertainty set in \eqref{eq:l1uncertaintyset}. To this end, for $\delta \in [0,1]$, define
\begin{equation}
\label{eq:betastar}
    \beta^*(\delta) \coloneqq \sqrt{\frac{2}{m}(\log(2^N-2)-\log \delta)}.
\end{equation}
By Theorem~\ref{thm:weissman}, $\beta^*(\delta)$ is the radius of an $\ell_1$-ball around $\hat P_X$ in which the true distribution $P_X$ will lie with probability at least $1-\delta$. We can use Theorem~\ref{thm:l1gamma} to bound the leakage inside of this set as $C\big(P_{Y|X},\mathcal B_{\beta^*(\delta)}(\hat P_X)\big) \leq \varepsilon'(\delta)$, where
\begin{equation}
    \varepsilon'(\delta) \coloneqq \varepsilon + \log \bigg(1-\frac{\beta^*(\delta)e^{\varepsilon}}{2}\bigg)^{-1}.
     \label{eq:epsmax}
\end{equation}
 That is, the mechanism will satisfy $\varepsilon'$-PML for all $P_X \in \mathcal B_{\beta^*(\delta)}(\hat P_X)$. If we now fix some $\varepsilon'>\varepsilon$ instead, we may use the same chain of argument to arrive at the following result.

\begin{corollary}
\label{corr:deltamin}
Assume $P_{Y|X}$ satisfies $\varepsilon$-PML with respect to $\hat P_X$ estimated form $m$ data samples such that the assumptions in Theorem~\ref{thm:l1gamma} hold. Then for any $\varepsilon'>\varepsilon$, the mechanism satisfies $(\varepsilon',\delta_{\min})$-PML, where,
    \begin{equation}
\label{eq:rrdeltamin}
    \delta_{\min}(\epsilon') \leq (2^N - 2)\exp\biggl(-2m(e^{-\epsilon}-e^{-\varepsilon'})^2\biggr).
\end{equation}
\end{corollary}
\begin{proof}
     Fix some $\varepsilon' > \varepsilon$. By rearranging \eqref{eq:epsmax}, we may find the value of $\beta^*(\delta)$ such that the leakage capacity on the set $\mathcal B_{\beta^*(\delta)}(\hat P_X)$ is upper bounded by $\varepsilon'$. Further solving for $\delta$ in \eqref{eq:betastar} yields the presented bound on $\delta_{\min}$.
\end{proof}

Note that Corollary~\ref{corr:deltamin} illustrates a few intuitive facts: data-generating distributions are harder to estimate for larger alphabet sizes; access to many data samples decreases the probability of failure of any $\varepsilon'$-PML guarantee; small deviations form the design-parameter $\epsilon$ are hard to achieve if only few data samples are available to estimate the data-generating distribution. For the case $N=20$, $\epsilon = \log(5)$, this behavior is shown in Figure~\ref{fig:deltaminplot:subfig:deltaeps}. Similarly, we can directly use \eqref{eq:epsmax} to determine the maximum increase in PML a mechanism incurs given a specified (upper bound on the) probability of failure $\delta$. This tradeoff between sample size $m$ and PML guarantee is depicted in Figure~\ref{fig:deltaminplot:subfig:epsdelta}. As can be seen from Figure~\ref{fig:deltaminplot}, the required sample size decreases quickly with small increases in the allowed variation in $\varepsilon$. In practice, this can be incorporated into mechanism design by first specifying a required $(\varepsilon',\delta)$-PML guarantee, and then designing a mechanism with a suitable $\varepsilon < \varepsilon'$ which yields the required guarantee.

\section{Optimal Mechanisms under $\ell_1$-uncertainty}
\label{sec:optimalbinary}
In this section, we will explore some results on optimality of extremal mechanisms in the distribution estimation scenario presented in Section~\ref{sec:distest}. 
We consider the following optimization problem: Let $\hat P_X$ be the distribution estimate \eqref{eq:empiricaldistribution}, and let $\beta$ be given by \eqref{eq:betastar}. With these quantities, define $\mathcal B_\beta(\hat P_X)$ as in \eqref{eq:l1uncertaintyset}. Then for some $\varepsilon \geq 0$, we aim to solve
\begin{align}
\label{eq:optprob}
    \max_{P_{Y|X} \in \mathcal M(\varepsilon,\mathcal B_\beta(\hat P_X))} &U(P_{Y|X},\hat P_X),
\end{align}
where $\mathcal M(\varepsilon,\mathcal P)$ is the set of all $N\times N$ mechanisms satisfying $\varepsilon$-PML with respect to all distributions in $\mathcal P$ in \eqref{eq:MmathcalP}, and $U$ is a sub-convex function as defined in \eqref{eq:subconvexdef}.
Firstly, we present a closed-form binary mechanism that maximizes sub-convex functions under this distributional uncertainty. The proof utilizes the techniques previously presented in \cite{grosse2023extremal}, and is a straightforward application of the results therein.
\begin{theorem}
\label{thm:optimalbinary}
    Let $\mathcal X \coloneqq \{x_1,x_2\}$. Given $\hat P_X\in \interior(\mathcal P_{\mathcal X})$ with $\hat{P}_X(x_1) \eqqcolon p_1 \geq p_2 \coloneqq \hat P_X(x_2)$, let $\beta < 2p_2$. Then for any $0 \leq \epsilon \leq -\log(p_1-\nicefrac{\beta}{2})$, $P^*_{Y|X}$ is a solution to \eqref{eq:optprob} with 
    \begin{equation}
        P^*_{Y|X} = \frac{1}{1+\beta e^{\epsilon}}\begin{bmatrix}
            e^{\epsilon}(1-p_1+\nicefrac{\beta}{2}) & 1-e^{\epsilon}(1-p_1-\nicefrac{\beta}{2}) \\
            1-e^{\epsilon}(p_1-\nicefrac{\beta}{2}) & e^{\epsilon}(p_1+\nicefrac{\beta}{2})
        \end{bmatrix}
    \end{equation}
\end{theorem}
\begin{proof}
 See Appendix~\ref{app:proofpotimalbinary}.
\end{proof}

\begin{remark}
\label{rem:binarymechisRR}
    For $p_1 = 0.5$ and $\beta = 1$, we have $\mathcal B_\beta(\hat P_X) = \mathcal P_{\mathcal X}$, that is, $\mathcal B_\beta(\hat P_X)$ contains all possible distributions on a binary alphabet. The mechanism then becomes
    \begin{equation}
        P^*_{Y|X} = \frac{1}{1+e^{\varepsilon}}\begin{bmatrix}
            e^{\varepsilon} & 1 \\
            1 & e^\varepsilon
        \end{bmatrix},
    \end{equation}
    which is exactly Warner's randomized response mechanism for local differential privacy \cite{warnerRRoriginal,extremalmechanismLong}. This is expected, since guaranteeing $\varepsilon$-PML for all possible data distributions is equivalent to guaranteeing $\varepsilon$-LDP \cite{saeidian2023pointwise,10619510}.
\end{remark}

Theorem~\ref{thm:optimalbinary} provides a closed form expression for the optimal binary mechanism given the $\ell_1$-radius of the uncertainty set. The condition $\varepsilon < -\log(p_1-\nicefrac{\beta}{2})$ ensures that $\varepsilon$ will remain in the \emph{first} privacy region for all data distributions in $\mathcal B_\beta(\hat P_X)$.\footnote{Note that for the remaining \emph{second} privacy region, the same proof technique can be applied by taking into account the box-constraint $0\leq P_{Y|X=x}(y)\leq 1$. We omit this case for the sake of clarity.} For general alphabets, and any value of $\varepsilon\geq 0$, the following result shows that the mechanism design problem can be solved by a linear constraint convex program.
\begin{proposition}
\label{prop:LCCP}
    Let $\hat P_X \in \interior(\mathcal P_{\mathcal X})$ and $\beta < 2\hat P_X(x)$ for all $x \in \mathcal X$. Let $\mathcal B_\beta(\hat P_X)$ be according to \eqref{eq:l1uncertaintyset}.
    For $i,j \in [N]$, denote $P_{Y|X=x_i}(y_j) \eqqcolon p_{ij}$.
    For any $\varepsilon\geq0$, the set $\mathcal M(\varepsilon,\mathcal B_\beta(\hat P_X))$ is a convex polytope in $[0,1]^{N\times N}$ given by the following linear constraints. 
\begin{IEEEeqnarray*}{lCr}
\label{eq:LPconst1}
        &p_{ij}  \leq e^{\varepsilon'}\bigg(\sum_{i=1}^N\hat P_X(x_i)p_{ij} + \frac{\beta} {2}p_{i'j}\bigg) \quad \forall i,i',j\in [N],& \yesnumber \\
\label{eq:LPconst2}
        &\sum_{j=1}^N p_{ij} = 1 \qquad \forall i \in [N],&\yesnumber\\
\label{eq:LPconst3}
        &p_{ij} \geq 0 \qquad \forall i,j \in [N].\yesnumber&
    \end{IEEEeqnarray*}
with the modified privacy parameter $\varepsilon'$ given by
\begin{equation}
\label{eq:optimalepsdash}
    \varepsilon' = \log \frac{e^{\varepsilon}}{1+\frac{\beta}{2}e^{\varepsilon}}.
\end{equation}
\end{proposition}
\begin{proof}
    See Appendix \ref{app:proofLCCP}.
\end{proof}

Proposition \ref{prop:LCCP} shows that the optimal mechanisms with respect to any convex utility function will be at one of the extreme points of the polytope $\mathcal M(\varepsilon,\mathcal B_\beta(\hat P_X))$ defined by the linear constraints \eqref{eq:LPconst1}-\eqref{eq:LPconst3}. We point out that with the additional condition $\min_{i\in[N]}p_{ij} = 0$ for all $j \in [N]$, the constraints become exactly equivalent to the constraints of the polytope of valid mechanisms satisfying $\varepsilon'$-PML with respect to the fixed data distribution $\hat P_X$ in \cite[Lemma 1]{grosse2023extremal}. This observation implies that with these additional constraints, the optimal mechanisms in $\mathcal M(\varepsilon',\hat P_X)$ presented in \cite{grosse2023extremal} are also optimal for $\mathcal M(\varepsilon,\mathcal B_\beta(\hat P_X))$, with $\varepsilon'$ obtained from the parameter transformation in \eqref{eq:optimalepsdash}. It can also be seen from \eqref{eq:Spx_l1bound} and \eqref{eq:optimalepsdash} that $\varepsilon'$ is exactly the upper bound on the leakage capacity for $\ell_1$-uncertainty sets in Theorem \ref{thm:l1gamma} in the case $k>1$. This implies
\begin{equation}
\label{eq:LPsubset}
    \mathcal M(\varepsilon',\hat P_X) \subseteq \mathcal M\Big(\varepsilon,\mathcal B_\beta(\hat P_X)\Big).
\end{equation}
That is, mechanisms designed to satisfy $\varepsilon'$-PML with respect to the distribution estimate $\hat P_X$ will satisfy $\varepsilon$-PML with respect to \emph{all} distributions in $\mathcal B_\beta(\hat P_X)$. Since $\mathcal M(\varepsilon',\hat P_X)$ may be a proper subset of $\mathcal M(\varepsilon,\mathcal B_\beta(\hat P_X))$, the optimality statements can not be directly transferred. However, as $\min_{i}p^*_{ij} \to 0$ for the extreme points $[p^*_{ij}]_{ij}$ of $\mathcal M(\varepsilon,\mathcal B_\beta(\hat P_X))$, the discrepancy between the optimal fixed-distribution mechanisms and the true extreme points will vanish.

\section{Additive Noise Mechanisms}
\label{sec:examples}
In this section, we design \emph{additive noise mechanisms} with the empirical estimate \eqref{eq:empiricaldistribution} computed from $m$ samples using the framework presented above. We will separately consider additive Laplace noise, and additive Gaussian noise in a \emph{local} model of noise addition, that is, the noise is added independently to each single data sample.
Throughout this section, let $X$ be a binary random variable defined on $\mathcal X = \{-1,1\}$. Consider a dataset $\{d_i\}_{i=1}^m$ where each $d_i$ is drawn independently according to a (unknown) distribution $P_X$. 

 As before, we will often need to translate between the probability of failure $\delta$ and the $\ell_1$-radius of the uncertainty sets in \eqref{eq:l1uncertaintyset}. Given any $\delta \in (0,1]$, for binary $X$, \eqref{eq:betastar} becomes
    \begin{equation}
    \label{eq:binarybeta}
        \beta^*(\delta) = \sqrt{\frac{2(\log 2 - \log \delta)}{m}}.
    \end{equation}
    
Assume that $Y$ is a random variable on the measure space $(\mathbb R,\mathfrak B(\mathbb R),\mathbb P)$, e.g., the random variable induced by adding Laplace or Gaussian noise to the binary random variable $X$. Then we use use $P_Y$ to denote the distribution of $Y$, and $P_{Y|X=x}$ to denote the distribution of $Y$ given that $X=x$. Since both Laplace and Gaussian noise are absolutely continuous with respect to the Lebesgue measure, we use $f_Y$ and $f_{Y|X=x}$ to denote the corresponding pdf induced by $P_Y$ and $P_{Y|X=x}$, respectively. For bounded densities, the leakage is
\begin{equation}
    \ell(X\to y) = \log \frac{\max_{x\in\mathcal X}f_{Y|X=x}(y)}{f_Y(y)}, \quad y\in \mathbb R.
\end{equation}
For a comprehensive treatment of PML on general probability spaces, see \cite{saeidian2023pointwisegeneral}. We begin by examining the addition of Laplace noise.
\subsection{Binary Laplace Mechanism}
\label{subsec:laplace}
 Consider the additive Laplace  mechanism defined by
 \begin{equation}
 \label{eq:laplacemech}
     Y = X + L,
 \end{equation}
 where $L$ is zero mean Laplace with scale parameter $b$, that is,
 \begin{equation}
     f_L(t) = \frac{1}{2b}\exp\bigg(-\frac{|t|}{b}\bigg), \quad t \in \mathbb R.
 \end{equation}
 For this fixed mechanism, we only need to find a suitable scale parameter $b$ such that the mechanism satisfies $(\epsilon,\delta)$-PML. We first find the relationship between $\varepsilon$ and $b$ for a fixed (and known) data distribution $P_X$.

\begin{proposition}
\label{prop:Laplaceleakage}
    Assume a binary random variable $X$ defined on $\mathcal{X} = \{-1,1\}$ to be distributed according to $P_X$ and let $p_{\min} \coloneqq \min_{x\in\{-1,1\}}P_X(x)$. The Laplace mechanism in \eqref{eq:laplacemech} satisfies $\mathcal \varepsilon(b)$-PML, where
    \begin{equation}
        \varepsilon(b) \coloneqq \varepsilon_{\min}(f_{Y|X},P_X) = \frac{2}{b} - \log(e^{\nicefrac{2}{b}}p_{\min} + 1-p_{\min}).
    \end{equation}
\end{proposition}
\begin{proof}
   See Appendix~\ref{app:proofLaplaceLeakage}.
\end{proof}

Observe that the value of $\varepsilon(b)$ is decreasing in $p_{\min}$. This implies that the worst-case leakage with respect to any distribution in an uncertainty set $\mathcal P$ will be the value of $\varepsilon(b)$ with respect to the distribution that contains the smallest minimum probability mass $p_{\min}$. If the uncertainty set is determined from \eqref{eq:l1uncertaintyset} with $\beta^*(\delta)$ from \eqref{eq:betastar}, we can therefore bound the leakage of the binary Laplace mechanism with respect to some $\hat P_X$ estimated from $m$ samples according to \eqref{eq:empiricaldistribution} and some $\delta$ as 
\begin{equation}
\label{eq:uncertainlaplace}
    \varepsilon(b,\delta) = \frac{2}{b} - \log \bigg[\bigg(\hat p_{\min}-\frac{\beta^*(\delta)}{2}\bigg)\big(e^{\nicefrac{2}{b}}-1\big)+1\bigg],
\end{equation}
with $\hat p_{\min} = \min_{x\in\mathcal X} \hat P_X(x)$ from the distribution estimate.

\begin{remark}
\label{rem:laplacereducestoldp}
    If $\hat p_{\min} = \nicefrac{1}{2}$ and $\beta^*(\delta) = 1$, then \eqref{eq:uncertainlaplace} yields
    \begin{equation}
        \varepsilon(b,\delta) = \nicefrac{2}{b}.
    \end{equation}
    That is, the $\varepsilon$-PML guarantee of the binary Laplace mechanism is equivalent to the $\varepsilon$-LDP guarantee of the mechanism \cite{dwork2014algorithmic}. As in Remark~\ref{rem:binarymechisRR}, observe that $\hat p_{\min} = \nicefrac{1}{2}$ implies that $\hat P_X = [\nicefrac{1}{2},\nicefrac{1}{2}]$, which together with $\beta^*(\delta) = 1$ implies that the $\mathcal \varepsilon(b)$-PML guarantee holds for all possible data distributions on the alphabet $\mathcal X$. Since satisfying $\varepsilon(b)$-PML for all data distributions is equivalent to satisfying $\varepsilon(b)$-LDP  \cite{saeidian2023pointwise,10619510}, this is a consequence of the formulation of approximate PML.
\end{remark}

\subsection{Gaussian Mechanism}
\label{subsec:gauss}
Recall that $X$ is a binary random variable defined on $\mathcal X = \{-1,1\}$. We consider the additive Gaussian mechanism
\begin{equation}
    Y = X + G,
\end{equation}
where $G \sim \mathcal N(0,\sigma^2)$ is zero-mean Gaussian, that is,
\begin{equation}
    f_G(t) = \frac{1}{\sqrt{2\pi\sigma^2}}\exp\bigg(-\frac{t^2}{2\sigma^2}\bigg), \quad t\in\mathbb R.
\end{equation}
First, assume that the true distribution $P_X$ is known. The PML with respect to this distribution for any outcome $y \in \mathbb R$ is
\begin{equation}
    \ell_{P_{Y|X}\times P_X}(X\to y) = \log \frac{\max_{x\in\mathcal X} f_{Y|X=x}(y)}{f_Y(y)},
\end{equation}
where the conditional distributions are given by
\begin{equation}
    f_{Y|X=x}(y) = \frac{1}{\sqrt{2\pi\sigma^2}}\exp\bigg(-\frac{1}{2\sigma^2}(y-x)^2\bigg),
\end{equation}
and $P_Y = P_{Y|X}\circ P_X$ is a Gaussian mixture with pdf
\begin{align}
    f_Y(y) = \frac{1}{\sqrt{2\pi\sigma^2}} \bigg[&P_X(x_1)\exp\bigg(-\frac{1}{2\sigma^2}(y+1)^2\bigg) \\[.5em] &+P_X(x_2)\exp\bigg(-\frac{1}{2\sigma^2}(y-1)^2\bigg)\bigg].
\end{align}
Next, we derive the achievable $(\varepsilon,\delta)$-PML guarantee using the estimator in \eqref{eq:empiricaldistribution} for the Gaussian mechanism: For any given $\varepsilon \geq 0$, we need to find corresponding $\delta_1 \geq 0$ such that the value $\ell(X\to y)$ is upper bounded by $\varepsilon$ with probability $1-\delta_1$ for any $P_X \in \mathcal B_\beta(\hat P_X)$. To this end, define 
\begin{equation}
    E_{\varepsilon}(P_X) \coloneqq \{y \in \mathbb R: \ell_{P_{Y|X}\times P_X}(X\to y) > \varepsilon\},
\end{equation}
as the set of outcomes $y$ that violate the $\varepsilon$-PML constraint. We want that the worst-case distribution in the uncertainty set violates this constraint with probability $\delta_1(\varepsilon)$. Thus we have,
\begin{align}    
    \delta_1(\varepsilon) &= \sup_{P_X\in\mathcal B_\beta(\hat P_X)}\mathbb P_{Y\sim P_{Y|X}\circ P_X}\big[E_\varepsilon(P_X)\big] \\[.5em]
    &= \sup_{P_X\in\mathcal B_\beta(\hat P_X)}\int_{E_\varepsilon(P_X)}d(P_{Y|X}\circ P_X)(y) \\[.5em]
    &=\sup_{P_X \in \mathcal B_\beta(\hat P_X)} \int_{E_\varepsilon(P_X)}f_Y(y)d\lambda(y) \label{eq:gaussianopt}\\[.5em]
    &= \sup_{P_X \in \mathcal B_\beta(\hat P_X)} \int_{E_\varepsilon(P_X)} \sum_{x \in\mathcal X}P_X(x)f_{Y|X=x}(y)d\lambda(y),
\end{align}
where $\lambda$ denotes the Lebesgue measure on $\mathbb R$. Define $\delta_2\geq 0$ as the maximum probability of failure tolerated for the distribution estimate $\hat P_{X}$ from $m$ samples using \eqref{eq:empiricaldistribution}. With the $\varepsilon$ defined above, we further use the result in Theorem~\ref{thm:l1gamma} to bound the leakage capacity of the mechanism by 
\begin{align}
    C(P_{Y|X},\mathcal B_\beta(\hat P_X)) &\leq \varepsilon + S_{\hat P_X}(P_{Y|X},\mathcal B_\beta(\hat P_X)) \\[.6em]& \leq \varepsilon - \log \bigg(1-\frac{\beta(\delta_2)}{2}e^{\varepsilon}\bigg) \eqqcolon \varepsilon^*
\end{align}
with probability $1-\delta_1$. By the union bound, the leakage of the binary Gaussian mechanism is hence bounded by $\varepsilon^*$ with probability $1-\delta^*$, where 
\begin{equation}
    \delta^* \coloneqq \delta_1 + \delta_2 - \delta_1\delta_2.
\end{equation}
The mechanism hence satisfies $\varepsilon^*$-PML with probability $1-\delta^*$.

\subsection{Numerical Experiments}
\begin{figure*}
    \centering
    \begin{subfigure}{.49\textwidth}
        \includegraphics[width=1\linewidth]{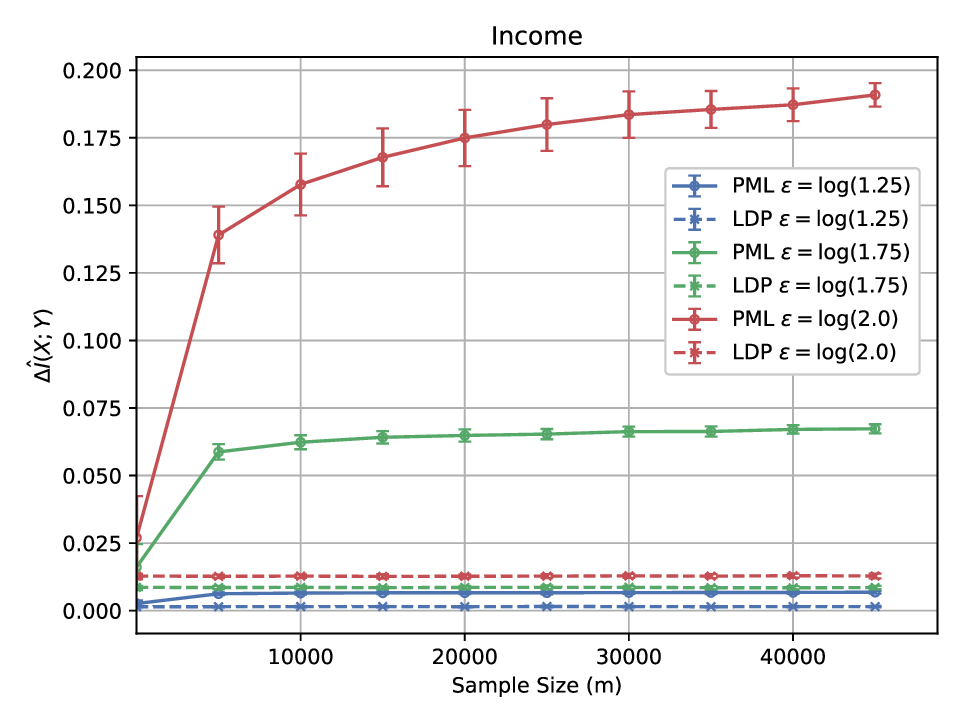}
    \end{subfigure}
    \begin{subfigure}{.49\textwidth}
        \includegraphics[width=1\textwidth]{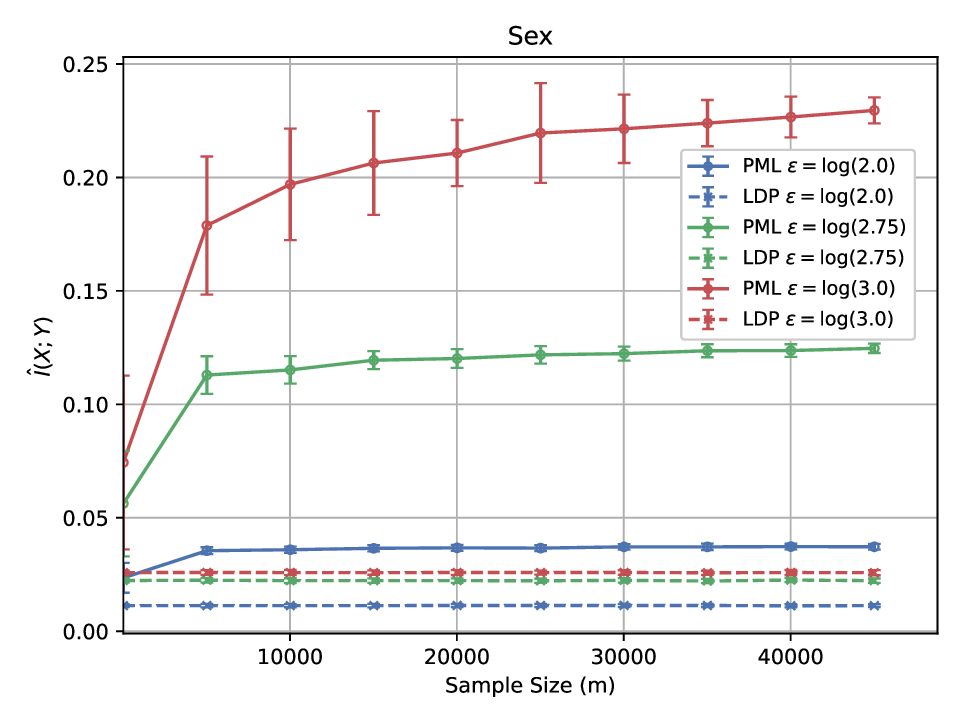}
    \end{subfigure}
    \caption{Comparison of the empirical mutual information utility between private and privatized data from the Adult dataset \cite{adult_2} as described in Section \ref{subsubsec:numex_laplace}. The probability of failure for the PML mechanism is set to $\delta = 10^{-9}$ for each simulation.}
    \label{fig:numex_laplace}
\end{figure*}
\begin{figure}
    \centering
    \includegraphics[width=1\linewidth]{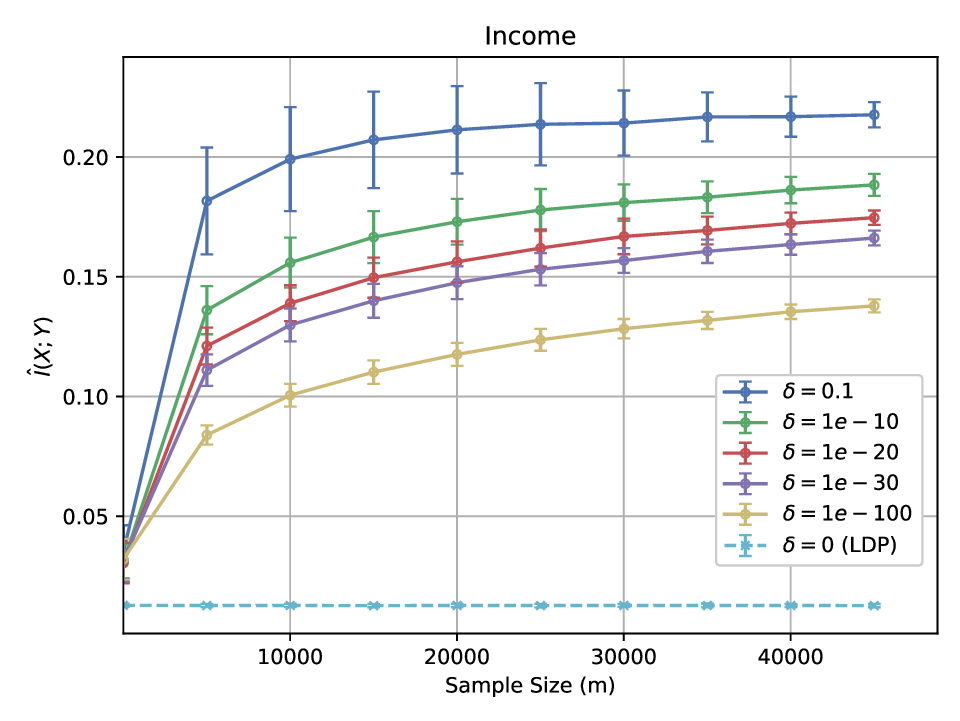}
    \caption{Comparison of empirical mutual information utility for different choices of the probability of failure $\delta$ in the binary Laplace mechanism according to Section \ref{subsec:laplace} with $\varepsilon = \log(2)$. For $\delta=0$, the Laplace mechanism for $(\varepsilon,\delta)$-PML reduces to the mechanism for $\varepsilon$-LDP (see Remark \ref{rem:laplacereducestoldp}).}
    \label{fig:laplaceoverdelta}
\end{figure}
In this section, we evaluate the presented additive noise mechanisms against similar mechanisms for local differential privacy. We use data from the Adult dataset \cite{adult_2} to evaluate the binary Laplace mechanism. Further, we numerically analyze the $(\varepsilon,\delta)$-curves for the presented Gaussian mechanism at a fixed noise variance, and compare them to the curve of mechanisms satisfying a localized version of probabilistic differential privacy as presented in \cite{zhao2019reviewingimprovinggaussianmechanism}.
\subsubsection{Laplace Mechanism}
\label{subsubsec:numex_laplace}
Using the two binary features \emph{income} and \emph{sex} of the Adult data set \cite{adult_2}, we implement the Laplace mechanism according to Section \ref{subsec:laplace}. We average the simulation over $100$ iterations. Each iteration, the dataset is randomly shuffled. Then the first $m$ samples of the shuffled dataset are used to compute $\hat P_X$. For different values of $\varepsilon$, the required scale parameter $b$ is determined based on \eqref{eq:uncertainlaplace}, where the probability of failure is set to $\delta = 10^{-9}$. For each $m$ and each value of $\varepsilon$, the data is then perturbed with Laplace noise given the calculated parameter. Let $\{d_j\}_{i=j}^m$ be the private data samples, and let $\{d'_j\}_{j=1}^m$ be the privatized samples obtained by the perturbation with Laplacian noise. Further, recall that $\{x_1,x_2\} = \{-1,1\}$, and based on this, also define $\{y_1,y_2\} \coloneqq \{-1,1\}$. To obtain binary features $\{d''_j\}_{j=1}^m$ form the privatized data samples $d'_j$ (which are arbitrary real numbers after the perturbation with the Laplacian noise), we implement a binary decision according to 
\begin{equation}
    d''_j \coloneqq \begin{cases}
        -1, &\text{ if }d'_j <0 \\
        1, &\text{ otherwise}
    \end{cases}, \quad \forall j = 1,\dots, m.
\end{equation}
To evaluate the statistical similarity between the private and privatized data, we define the empirical mutual information as 
\begin{equation}
   \hat{I}(X;Y) \coloneqq \sum_{i=1}^2 \sum_{k=1}^2 \frac{f(x_i,y_k)}{m}\log\frac{Nf(x_i,y_k)}{f(x_i)f(y_k)},
\end{equation}
where $f(x_i,y_k)$ denotes the frequency of the tuple $(x_i,y_k)$ in the sequence $\{(d_1,d''_1),\dots,(d_m,d_m'')\}$, and $f(x_i)$ and $f(y_k)$ denote the frequencies of symbols $x_i$ and $y_j$ in the sequences $\{d_j\}_{j=1}^m$ and $\{d_j''\}_{j=1}^m$, respectively.  
In Figure \ref{fig:numex_laplace}, we compare the empirical mutual information between the true and perturbed data points with the noise optimized for PML to the empirical mutual information between the true and perturbed data when the Laplace noise is chosen to satisfy $\varepsilon$-LDP (that is, $b=\nicefrac{2}{\varepsilon}$). For both features, it becomes clear that incorporating the additional information about the data distribution of the data following the presented procedure allows us to significantly increase the data utility compared to LDP at the same privacy parameter. This comes only at the cost of the extremely small probability of failure $\delta = 10^{-9}$. As expected, the utility gain increases with the number of samples $m$ available for distribution estimation.

To evaluate the sensitivity of the utility increase to the chosen probability of failure $\delta$, we fix $\varepsilon$ and compare the empirical mutual information over the available data samples $m$ for different choices of $\delta$ in Figure \ref{fig:laplaceoverdelta}. As expected, the achievable utility decreases with decreasing choices of $\delta$. However, even for extremely small values of $\delta$, the utility increase in the presented binary setting is substantial compared to LDP (equivalent to choosing $\delta=0$, see Remark \ref{rem:laplacereducestoldp}).

\subsubsection{Gaussian Mechanism}
\begin{figure*}
\begin{subfigure}{.49\textwidth}
    \centering
    \includegraphics[width=1.03\linewidth]{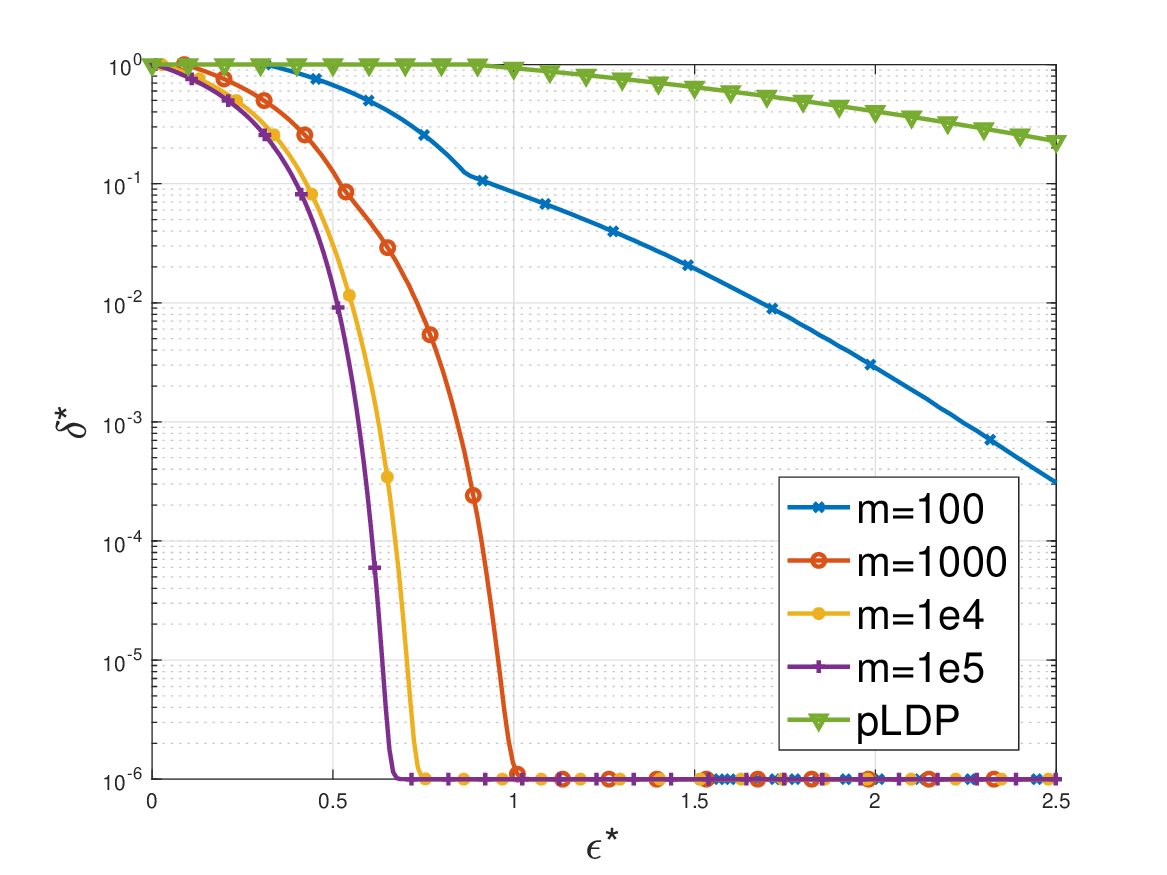}
    \caption{$\sigma = 1.5$}
    \label{fig:gaussexample:subfig:1p5}
\end{subfigure}
\begin{subfigure}{.49\textwidth}
    \centering
    \includegraphics[width=1.03\linewidth]{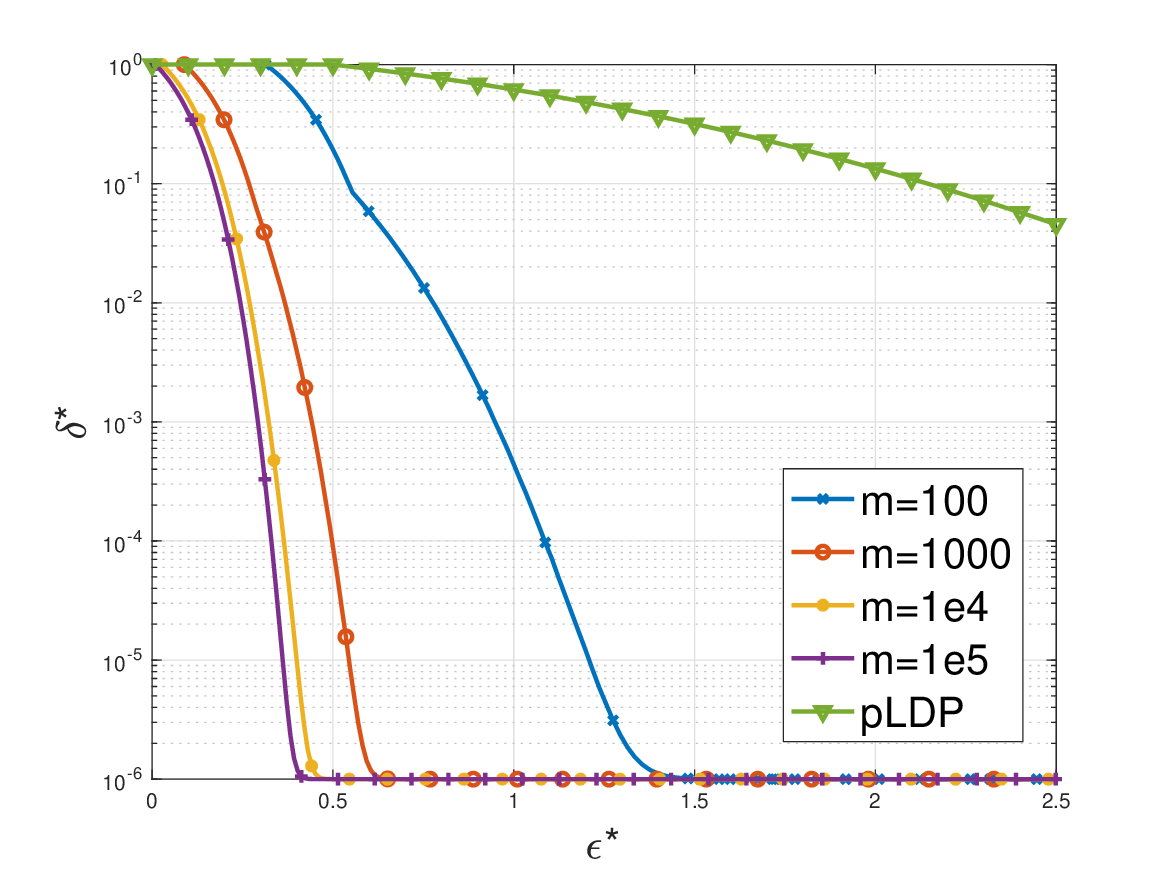}
    \caption{$\sigma = 2$}
    \label{fig:gaussexample:subfig:2}
\end{subfigure}
    \caption{$(\varepsilon^*,\delta^*)$-curves for the binary Gaussian mechanism in Section~\ref{subsec:gauss} given that $\hat P_X = (0.5,0.5)$. The radius $\beta$ is picked such that for each $m$, $\delta_2 = 10^{-6}$ according to \eqref{eq:binarybeta}. The optimization in \eqref{eq:gaussianopt} was solved numerically by a exhaustive grid search over the set $\mathcal B_\beta(\hat P_X)$. The lower bound $\delta^* \geq \delta_2$ is visible for all chosen $m \neq 100$. For the chosen $\sigma$, the curves are compared to the binary $(\varepsilon^*,\delta^*)$-pLDP mechanism, a local adaption of the mechanism given in \cite[Theorem 8]{zhao2019reviewingimprovinggaussianmechanism}.}
    \label{fig:gaussexample}
\end{figure*}
\begin{figure}
    \centering
    \includegraphics[width=1.03\linewidth]{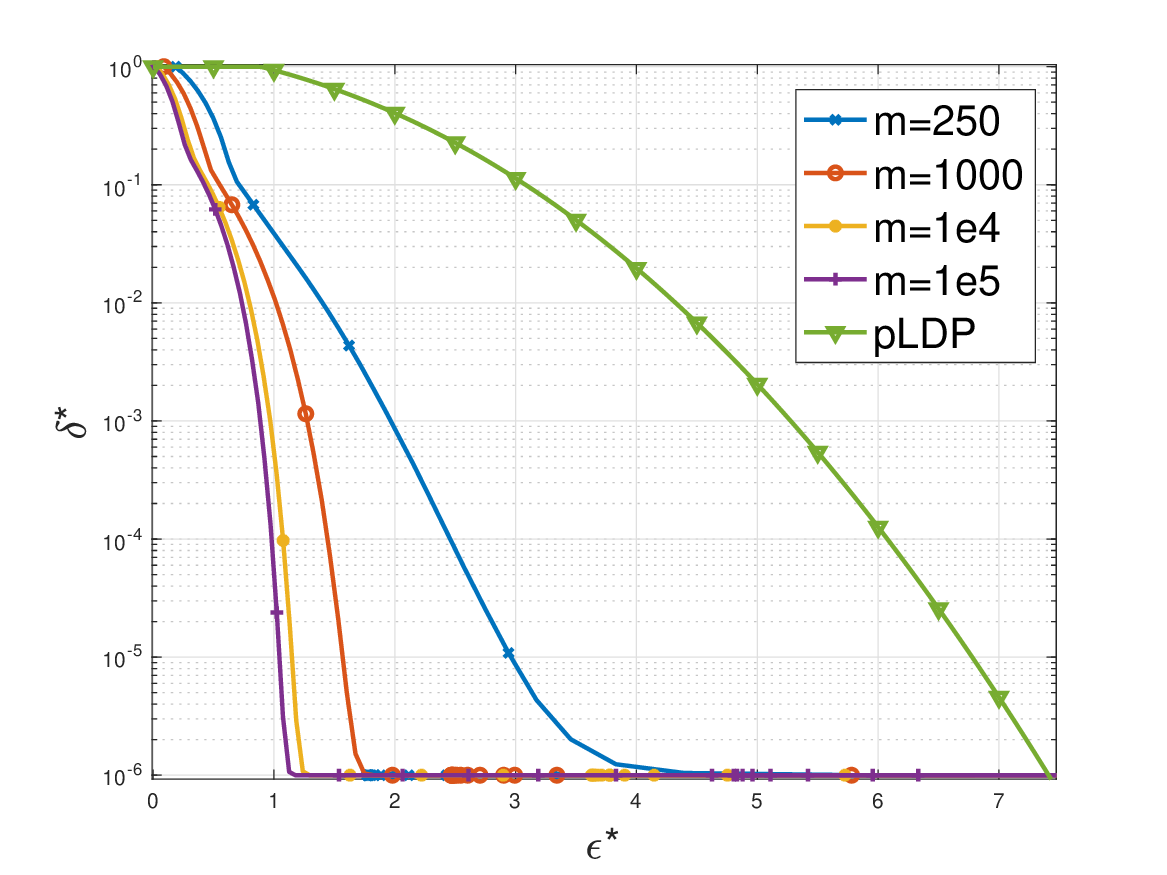}
    \caption{The same simulation as in Figure \ref{fig:gaussexample}, but given a skewed estimated distribution $\hat P_X = (0.3,0.7)$, with $\sigma = 1.5$.}
    \label{fig:skewedgauss}
\end{figure}
 Next, we evaluate the Gaussian mechanism, and compare it to Gaussian mechanisms for a probabilistic version of $(\varepsilon,\delta)$-LDP. Note in the usual way of defining $(\varepsilon,\Bar \delta)$-LDP, the parameter $\Bar \delta$ has a different interpretation than the $\delta$ in the framework presented above. In particular, if a mechanism $P_{Y|X}$ satisfies $(\varepsilon,\Bar \delta)$-LDP for some $\varepsilon>0$ and $\Bar \delta \in [0,1]$, this implies that
 \begin{equation}
 \label{eq:LDPdelta}
     P_{Y|X=x}(y) \leq e^{\varepsilon}P_{Y|X=x'}(y) + \Bar \delta \quad \forall x,x',y.
 \end{equation}
 While the probability of failure $\delta$ of an $\varepsilon$-LDP guarantee is an upper bound to the $\Bar \delta$ defined according to \eqref{eq:LDPdelta}, this upper bound is often not tight \cite{meiser2018approximate}, hindering a direct comparison. We therefore compare the $(\varepsilon,\delta)$-curves of the Gaussian mechanisms for PML to the $(\varepsilon,\delta)$-curve of the mechanism presented in \cite[Theorem 8]{zhao2019reviewingimprovinggaussianmechanism}: There, the authors define a Gaussian mechanism satisfying $(\varepsilon,\delta)$-pDP, that is, an additive Gaussian noise mechanism for which an $\varepsilon$-DP guarantee fails with probability at most $\delta$. We adapt this mechanism to our binary local setting by choosing the sensitivity $\Delta = 2$.\footnote{To see why this reduces the central setting to our local setting, for any two data points $d$ and $d'$, let $D = \{d\}$, $D' =\{d'\}$ and define the \say{query}$q: D \to \mathbb R$ to be the identity function, that is $q(D)=D$. It follows that $\max_{D\sim D'}||q(D)-q(D')||_2 = \max_{x,x'}||x-x'||_2 =2$, and bounding the ratio of outcomes for each pair of neighboring databases is equivalent to bounding that ratio for any two data points $x,x'$.} The mechanism in \cite[Theorem 8]{zhao2019reviewingimprovinggaussianmechanism} then satisfies \emph{$(\varepsilon,\delta)\text{-}$probabilistic LDP} (pLDP), that is, 
 \begin{equation}
     \mathbb P_{Y}\Bigg[\max_{x,x'}\frac{P_{Y|X=x}}{P_{Y|X=x'}} > e^{\varepsilon}\Bigg] < \delta,
 \end{equation}
 whenever we choose $G \sim \mathcal N(0,\sigma_{\text{pLDP}}^2)$, with
 \begin{equation}
     \sigma_{\text{pLDP}} = \frac{\sqrt{2}\big(f + \sqrt{f^2 + \varepsilon^*}\big)}{\varepsilon^*},
 \end{equation}
 where $f \coloneqq \text{inverfc}(\delta)$ is the inverse of the complementary error function at value $\delta$. The resulting $(\varepsilon^*,\delta^*)$ tradeoff computed using the technique in Section \ref{subsec:gauss} is compared to the corresponding pLDP mechanism in Figure~\ref{fig:gaussexample} for different values of $m$. We see that incorporating the distribution estimate in the privacy guarantee strongly improves the achievable $(\varepsilon,\delta)$-tradeoffs. In practice, this implies that a similar privacy guarantee to $(\varepsilon,\delta)$-probabilistic local differential privacy can be achieved with the presented techniques using lower noise variance, hence directly increasing data utility. Further, the difference between Figure \ref{fig:gaussexample:subfig:1p5} and Figrue \ref{fig:skewedgauss} once again illustrates the fact that an LDP guarantee is equivalent to a PML guarantee for all possible priors, with the worst-case privacy loss realized for the degenerate distributions with all probability mass concentrated at one outcome in the limit. Hence, as the true data-generating distribution $P_X$ becomes more skewed, the utility gain with the PML method presented in this paper decreases.

\section{Related Work}
\label{sec:relatedwork}
Information-theoretic approaches have played a central role in the development of alternative privacy measures to (L)DP. Numerous studies deal with the use of mutual information as a privacy measure~\cite{asoodeh2014notes,wang2016relation,makhdoumi2014information,liao2017hypothesis,rassouli2021perfect}. Beyond mutual information, generalized measures such as $f$-information have been explored, with specific instances including $\chi^2$-information~\cite{du2017principal,wang2019privacy} and total variation privacy~\cite{rassouli2019TVDprivacy}. Per-letter measures based on $f$-divergences have also been proposed for assessing privacy~\cite{zamani2021data,zamani2021design,zamani2023privacy}, along with measures based on guessing probability~\cite{asoodeh2018estimation} or guesswork~\cite{massey1994guessing,malone2004guesswork}. As an analouge to LDP, local information privacy (LIP) \cite{6483382,jiang2021LIPcontextaware} imposes symmetric bounds on the information density, while asymmetric local information privacy (ALIP)\cite{zarrabian2023lift,zarrabian2022asymmetric} extends LIP by allowing different upper and lower bounds. Other notable advances include tunable privacy measures, such as $\alpha$-leakage and maximal $\alpha$-leakage~\cite{liao2019tunable}, which extend maximal leakage by introducing a class of loss function parameterized by $\alpha$. More recent developments in this area include maximal $(\alpha, \beta)$-leakage~\cite{gilani2024unifying}, and maximal $g$-leakage for multiple guesses~\cite{kurri2023operational}.

In the QIF framework, recent works examine the effect of imperfect statistical information \emph{at the adversary} \cite{sakib2023variations,sakib2023measures,sakib2024information}, alternative privacy measures based on cryptographic advantage \cite{chatzikokolakis2023bayes}, the properties of methods in federated learning \cite{biswas2024bayes}, as well as the shuffle model for differential privacy \cite{jurado2023analyzing}. Connections between LDP and QIF are explored in \cite{10664297}. Privacy in general learning problems is investigated in, e.g., \cite{del2023bounding,aubinais2023fundamental,9611429,9531956}. 
Finally, a rich body of work has focused on the design of privacy mechanisms tailored to specific privacy measures. For example, bespoke mechanisms have been developed for local differential privacy~\cite{kalantari2018hamming,extremalmechanismLong}, maximal leakage~\cite{wu2020optimal,saeidian2021hamming}, total variation privacy~\cite{rassouli2019TVDprivacy}, and LIP~\cite{hsu2019information,sadeghi2021properties,jiang2021LIPcontextaware}. In \cite{lopuhaa2024mechanisms,goseling2022robust}, mechanisms are designed for \emph{robust local differential privacy}, a specific instance of Pufferfish privacy \cite{kiferPufferfishFrameworkMathematical2014}.

\section{Conclusion and Discussions}
\label{sec:conclusion}
This paper provides a comprehensive framework for dealing with data distribution uncertainty in PML privacy assessment and mechanism design. To the best of our knowledge, this constitutes the first work that demonstrates how the data-distribution dependent PML measure can be used without assuming the data-generating distribution to be known. The theoretical tools developed in Section \ref{sec:PMLforsets} can be used to assess leakage if there is \emph{any} knowledge about the set which the data-generating distribution is drawn from. This can be used to obtain guarantees that hold for any kind of distributional uncertainty (including no knowledge of the data distribution). In the following sections, we focused our analysis on the specific case in which uncertainty stems from empirical distribution estimation. Specializing to this use-case enables us to derive (a) more concrete bounds on the leakage increase that can be used to obtain high-probability $\varepsilon$-PML guarantees, and (b) methods for computing optimal mechanisms for sub-convex utility. Finally, the results in Section \ref{sec:examples} show how the presented framework can be applied to additive noise mechanisms. We show numerically that the incorporation of distribution estimates in the privacy guarantees according to our framework can directly improve data utility when compared to the distribution-oblivious LDP approach. While the presented examples are minimal in terms of cardinality, the substantial utility gain they achieve compared to LDP should motivate further research into making these methods practically available whenever sufficient data is available for distribution estimation.


\textbf{Future work}: It is worth pointing out that the considered uncertainty sets from distribution estimation in Theorem \ref{thm:l1gamma} are just one practically relevant example. Other large deviation bounds (e.g. in $\ell_\infty$) might be used to derive similar results for different use-cases by the same methodology. Other ways to incorporate domain knowledge could include parameterized models, e.g., normal with unknown mean and variance. Analyzing the properties of the uncertainty set under such assumptions would be a first step to obtain effective mechanisms in these scenarios. More broadly, the exploration of mechanisms tailored towards a specific application, e.g., privacy guarantees provided by quantization noise or the generalization capabilities of a learning system, in the presented framework can potentially lead to alternative randomization strategies that can be used in these specific contexts, potentially leading to synergies with other parts of a data-processing pipeline.

\bibliographystyle{IEEEtranN}
\footnotesize
\bibliography{main}

\begin{thebibliography}{72}
\providecommand{\natexlab}[1]{#1}
\providecommand{\url}[1]{#1}
\csname url@samestyle\endcsname
\providecommand{\newblock}{\relax}
\providecommand{\bibinfo}[2]{#2}
\providecommand{\BIBentrySTDinterwordspacing}{\spaceskip=0pt\relax}
\providecommand{\BIBentryALTinterwordstretchfactor}{4}
\providecommand{\BIBentryALTinterwordspacing}{\spaceskip=\fontdimen2\font plus
\BIBentryALTinterwordstretchfactor\fontdimen3\font minus \fontdimen4\font\relax}
\providecommand{\BIBforeignlanguage}[2]{{%
\expandafter\ifx\csname l@#1\endcsname\relax
\typeout{** WARNING: IEEEtranN.bst: No hyphenation pattern has been}%
\typeout{** loaded for the language `#1'. Using the pattern for}%
\typeout{** the default language instead.}%
\else
\language=\csname l@#1\endcsname
\fi
#2}}
\providecommand{\BIBdecl}{\relax}
\BIBdecl

\bibitem[Narayanan and Shmatikov(2008)]{deanonymNetflix}
A.~Narayanan and V.~Shmatikov, ``Robust de-anonymization of large sparse datasets,'' in \emph{IEEE Symposium on Security and Privacy (SP)}, 2008, pp. 111--125.

\bibitem[Alvim et~al.(2022)Alvim, Fernandes, McIver, Morgan, and Nunes]{Alvim2022FlexibleAS}
\BIBentryALTinterwordspacing
M.~S. Alvim, N.~Fernandes, A.~McIver, C.~Morgan, and G.~H. Nunes, ``Flexible and scalable privacy assessment for very large datasets, with an application to official governmental microdata,'' \emph{Proc. Priv. Enhancing Technol.}, vol. 2022, pp. 378--399, 2022. [Online]. Available: \url{https://api.semanticscholar.org/CorpusID:248476260}
\BIBentrySTDinterwordspacing

\bibitem[Duchi et~al.(2013)Duchi, Jordan, and Wainwright]{duchi2013LDPminmaxDEF}
J.~C. Duchi, M.~I. Jordan, and M.~J. Wainwright, ``Local privacy and statistical minimax rates,'' in \emph{IEEE 54th Annual Symposium on Foundations of Computer Science}, 2013, pp. 429--438.

\bibitem[{Apple Differential Privacy Team}(2017)]{appleDP}
\BIBentryALTinterwordspacing
{Apple Differential Privacy Team}, ``Learning with privacy at scale,'' 2017. [Online]. Available: \url{https://docs-assets.developer.apple.com/ml-research/papers/learning-with-privacy-at-scale.pdf}
\BIBentrySTDinterwordspacing

\bibitem[Erlingsson et~al.(2014)Erlingsson, Pihur, and Korolova]{googleRAPPOR}
{\'U}.~Erlingsson, V.~Pihur, and A.~Korolova, ``Rappor: Randomized aggregatable privacy-preserving ordinal response,'' in \emph{Proceedings of the 2014 ACM SIGSAC conference on computer and communications security}, 2014, pp. 1054--1067.

\bibitem[Smith(2009)]{smith2009foundations}
G.~Smith, ``On the foundations of quantitative information flow,'' in \emph{International Conference on Foundations of Software Science and Computational Structures}.\hskip 1em plus 0.5em minus 0.4em\relax Springer, 2009, pp. 288--302.

\bibitem[Alvim et~al.(2020)Alvim, Chatzikokolakis, McIver, Morgan, Palamidessi, and Smith]{alvim2020science}
M.~S. Alvim, K.~Chatzikokolakis, A.~McIver, C.~Morgan, C.~Palamidessi, and G.~Smith, \emph{The Science of Quantitative Information Flow}.\hskip 1em plus 0.5em minus 0.4em\relax Springer Cham, 2020.

\bibitem[Braun et~al.(2009)Braun, Chatzikokolakis, and Palamidessi]{braun2009quantitative}
C.~Braun, K.~Chatzikokolakis, and C.~Palamidessi, ``Quantitative notions of leakage for one-try attacks,'' \emph{Electronic Notes in Theoretical Computer Science}, vol. 249, pp. 75--91, 2009.

\bibitem[Alvim et~al.(2012)Alvim, Chatzikokolakis, Palamidessi, and Smith]{alvim2012measuring}
M.~S. Alvim, K.~Chatzikokolakis, C.~Palamidessi, and G.~Smith, ``Measuring information leakage using generalized gain functions,'' in \emph{IEEE 25th Computer Security Foundations Symposium}, 2012, pp. 265--279.

\bibitem[Issa et~al.(2020)Issa, Wagner, and Kamath]{IssaMaxL}
I.~Issa, A.~B. Wagner, and S.~Kamath, ``An operational approach to information leakage,'' \emph{IEEE Transactions on Information Theory}, vol.~66, no.~3, pp. 1625--1657, 2020.

\bibitem[Saeidian et~al.(2023{\natexlab{a}})Saeidian, Cervia, Oechtering, and Skoglund]{saeidian2023pointwise}
S.~Saeidian, G.~Cervia, T.~J. Oechtering, and M.~Skoglund, ``Pointwise maximal leakage,'' \emph{IEEE Transactions on Information Theory}, vol.~69, no.~12, pp. 8054--8080, 2023.

\bibitem[Espinoza and Smith(2013)]{espinoza2013min}
B.~Espinoza and G.~Smith, ``Min-entropy as a resource,'' \emph{Information and Computation}, vol. 226, pp. 57--75, 2013.

\bibitem[Saeidian et~al.(2025)Saeidian, Cervia, Oechtering, and Skoglund]{saeidian2023inferential}
S.~Saeidian, G.~Cervia, T.~J. Oechtering, and M.~Skoglund, ``Rethinking disclosure prevention with pointwise maximal leakage,'' \emph{Journal of Privacy and Confidentiality}, vol.~15, no.~1, Mar. 2025.

\bibitem[Fernandes et~al.(2024)Fernandes, McIver, and Sadeghi]{10664297}
N.~Fernandes, A.~McIver, and P.~Sadeghi, ``Explaining $\epsilon$ in local differential privacy through the lens of quantitative information flow,'' in \emph{2024 IEEE 37th Computer Security Foundations Symposium (CSF)}, 2024, pp. 419--432.

\bibitem[Grosse et~al.(2024{\natexlab{a}})Grosse, Saeidian, and Oechtering]{grosse2023extremal}
L.~Grosse, S.~Saeidian, and T.~J. Oechtering, ``Extremal mechanisms for pointwise maximal leakage,'' \emph{IEEE Transactions on Information Forensics and Security}, vol.~19, pp. 7952--7967, 2024.

\bibitem[Domingo-Ferrer et~al.(2021)Domingo-Ferrer, S\'{a}nchez, and Blanco-Justicia]{10.1145/3433638}
\BIBentryALTinterwordspacing
J.~Domingo-Ferrer, D.~S\'{a}nchez, and A.~Blanco-Justicia, ``The limits of differential privacy (and its misuse in data release and machine learning),'' \emph{Commun. ACM}, vol.~64, no.~7, p. 33–35, Jun. 2021. [Online]. Available: \url{https://doi.org/10.1145/3433638}
\BIBentrySTDinterwordspacing

\bibitem[Bagdasaryan et~al.(2019)Bagdasaryan, Poursaeed, and Shmatikov]{10.5555/3454287.3455674}
E.~Bagdasaryan, O.~Poursaeed, and V.~Shmatikov, \emph{Differential privacy has disparate impact on model accuracy}.\hskip 1em plus 0.5em minus 0.4em\relax Red Hook, NY, USA: Curran Associates Inc., 2019.

\bibitem[Polyanskiy and Wu(2025)]{Polyanskiy_Wu_2025}
Y.~Polyanskiy and Y.~Wu, \emph{Information Theory: From Coding to Learning}.\hskip 1em plus 0.5em minus 0.4em\relax Cambridge University Press, 2025.

\bibitem[Rabbani and Joshi(2002)]{rabbani2002overview}
M.~Rabbani and R.~Joshi, ``An overview of the jpeg 2000 still image compression standard,'' \emph{Signal processing: Image communication}, vol.~17, no.~1, pp. 3--48, 2002.

\bibitem[Diaz et~al.(2020)Diaz, Wang, Calmon, and Sankar]{diaz2019robustness}
M.~Diaz, H.~Wang, F.~P. Calmon, and L.~Sankar, ``On the robustness of information-theoretic privacy measures and mechanisms,'' \emph{IEEE Transactions on Information Theory}, vol.~66, no.~4, pp. 1949--1978, 2020.

\bibitem[Liao et~al.(2019)Liao, Kosut, Sankar, and du~Pin~Calmon]{liao2019tunable}
J.~Liao, O.~Kosut, L.~Sankar, and F.~du~Pin~Calmon, ``Tunable measures for information leakage and applications to privacy-utility tradeoffs,'' \emph{IEEE Transactions on Information Theory}, vol.~65, no.~12, pp. 8043--8066, 2019.

\bibitem[Lopuhaä-Zwakenberg and Goseling(2024)]{lopuhaa2024mechanisms}
M.~Lopuhaä-Zwakenberg and J.~Goseling, ``Mechanisms for robust local differential privacy,'' \emph{Entropy}, vol.~26, no.~3, 2024.

\bibitem[Dwork et~al.(2014)Dwork, Roth, et~al.]{dwork2014algorithmic}
C.~Dwork, A.~Roth \emph{et~al.}, ``The algorithmic foundations of differential privacy,'' \emph{Foundations and Trends{\textregistered} in Theoretical Computer Science}, vol.~9, no. 3--4, pp. 211--407, 2014.

\bibitem[R{\'e}nyi(1961)]{renyi1961entropy}
A.~R{\'e}nyi, ``On measures of entropy and information,'' in \emph{Fourth Berkeley Symposium on Mathematical Statistics and Probability, Volume 1: Contributions to the Theory of Statistics}, vol.~4.\hskip 1em plus 0.5em minus 0.4em\relax University of California Press, 1961, pp. 547--562.

\bibitem[Kairouz et~al.(2016)Kairouz, Oh, and Viswanath]{extremalmechanismLong}
P.~Kairouz, S.~Oh, and P.~Viswanath, ``Extremal mechanisms for local differential privacy,'' \emph{The Journal of Machine Learning Research}, vol.~17, no.~1, pp. 492--542, 2016.

\bibitem[Weissman et~al.(2003)Weissman, Ordentlich, Seroussi, Verdu, and Weinberger]{weissman2003inequalities}
T.~Weissman, E.~Ordentlich, G.~Seroussi, S.~Verdu, and M.~J. Weinberger, ``Inequalities for the l1 deviation of the empirical distribution,'' \emph{Hewlett-Packard Labs, Tech. Rep}, p. 125, 2003.

\bibitem[Boyd and Vandenberghe(2014)]{boyd2004convex}
S.~P. Boyd and L.~Vandenberghe, \emph{Convex Optimization}.\hskip 1em plus 0.5em minus 0.4em\relax Cambridge University Press, 2014.

\bibitem[Sreekumar et~al.(2023)Sreekumar, Goldfeld, and Kato]{sreekumar2023limit}
S.~Sreekumar, Z.~Goldfeld, and K.~Kato, ``Limit distribution theory for kl divergence and applications to auditing differential privacy,'' in \emph{2023 IEEE International Symposium on Information Theory (ISIT)}.\hskip 1em plus 0.5em minus 0.4em\relax IEEE, 2023, pp. 2607--2612.

\bibitem[Altu{\u{g}} and Wagner(2014)]{altuug2014singulardistributions}
Y.~Altu{\u{g}} and A.~B. Wagner, ``Refinement of the random coding bound,'' \emph{IEEE Transactions on Information Theory}, vol.~60, no.~10, pp. 6005--6023, 2014.

\bibitem[Warner(1965)]{warnerRRoriginal}
S.~L. Warner, ``Randomized response: A survey technique for eliminating evasive answer bias,'' \emph{Journal of the American Statistical Association}, vol.~60, no. 309, pp. 63--69, 1965.

\bibitem[Grosse et~al.(2024{\natexlab{b}})Grosse, Saeidian, Sadeghi, Oechtering, and Skoglund]{10619510}
L.~Grosse, S.~Saeidian, P.~Sadeghi, T.~J. Oechtering, and M.~Skoglund, ``Quantifying privacy via information density,'' in \emph{2024 IEEE International Symposium on Information Theory (ISIT)}, 2024, pp. 3071--3076.

\bibitem[Saeidian et~al.(2023{\natexlab{b}})Saeidian, Cervia, Oechtering, and Skoglund]{saeidian2023pointwisegeneral}
S.~Saeidian, G.~Cervia, T.~J. Oechtering, and M.~Skoglund, ``Pointwise maximal leakage on general alphabets,'' in \emph{IEEE International Symposium on Information Theory (ISIT)}, 2023, pp. 388--393.

\bibitem[Becker and Kohavi(1996)]{adult_2}
B.~Becker and R.~Kohavi, ``{Adult},'' UCI Machine Learning Repository, 1996, {DOI}: https://doi.org/10.24432/C5XW20.

\bibitem[Zhao et~al.(2019)Zhao, Wang, Bai, Lam, Xu, Shi, Ren, Yang, Liu, and Yu]{zhao2019reviewingimprovinggaussianmechanism}
\BIBentryALTinterwordspacing
J.~Zhao, T.~Wang, T.~Bai, K.-Y. Lam, Z.~Xu, S.~Shi, X.~Ren, X.~Yang, Y.~Liu, and H.~Yu, ``Reviewing and improving the gaussian mechanism for differential privacy,'' 2019. [Online]. Available: \url{https://arxiv.org/abs/1911.12060}
\BIBentrySTDinterwordspacing

\bibitem[Meiser(2018)]{meiser2018approximate}
S.~Meiser, ``Approximate and probabilistic differential privacy definitions,'' \emph{Cryptology ePrint Archive}, 2018.

\bibitem[Asoodeh et~al.(2014)Asoodeh, Alajaji, and Linder]{asoodeh2014notes}
S.~Asoodeh, F.~Alajaji, and T.~Linder, ``Notes on information-theoretic privacy,'' in \emph{52nd Annual Allerton Conference on Communication, Control, and Computing (Allerton)}, 2014, pp. 1272--1278.

\bibitem[Wang et~al.(2016)Wang, Ying, and Zhang]{wang2016relation}
W.~Wang, L.~Ying, and J.~Zhang, ``On the relation between identifiability, differential privacy, and mutual-information privacy,'' \emph{IEEE Transactions on Information Theory}, vol.~62, no.~9, pp. 5018--5029, 2016.

\bibitem[Makhdoumi et~al.(2014)Makhdoumi, Salamatian, Fawaz, and M{\'e}dard]{makhdoumi2014information}
A.~Makhdoumi, S.~Salamatian, N.~Fawaz, and M.~M{\'e}dard, ``From the information bottleneck to the privacy funnel,'' in \emph{IEEE Information Theory Workshop (ITW)}, 2014, pp. 501--505.

\bibitem[Liao et~al.(2017)Liao, Sankar, Tan, and du~Pin~Calmon]{liao2017hypothesis}
J.~Liao, L.~Sankar, V.~Y. Tan, and F.~du~Pin~Calmon, ``Hypothesis testing under mutual information privacy constraints in the high privacy regime,'' \emph{IEEE Transactions on Information Forensics and Security}, vol.~13, no.~4, pp. 1058--1071, 2017.

\bibitem[Rassouli and G{\"u}nd{\"u}z(2021)]{rassouli2021perfect}
B.~Rassouli and D.~G{\"u}nd{\"u}z, ``On perfect privacy,'' \emph{IEEE Journal on Selected Areas in Information Theory}, vol.~2, no.~1, pp. 177--191, 2021.

\bibitem[du~Pin~Calmon et~al.(2017)du~Pin~Calmon, Makhdoumi, M{\'e}dard, Varia, Christiansen, and Duffy]{du2017principal}
F.~du~Pin~Calmon, A.~Makhdoumi, M.~M{\'e}dard, M.~Varia, M.~Christiansen, and K.~R. Duffy, ``Principal inertia components and applications,'' \emph{IEEE Transactions on Information Theory}, vol.~63, no.~8, pp. 5011--5038, 2017.

\bibitem[Wang et~al.(2019)Wang, Vo, Calmon, M{\'e}dard, Duffy, and Varia]{wang2019privacy}
H.~Wang, L.~Vo, F.~P. Calmon, M.~M{\'e}dard, K.~R. Duffy, and M.~Varia, ``Privacy with estimation guarantees,'' \emph{IEEE Transactions on Information Theory}, vol.~65, no.~12, pp. 8025--8042, 2019.

\bibitem[Rassouli and G{\"u}nd{\"u}z(2019)]{rassouli2019TVDprivacy}
B.~Rassouli and D.~G{\"u}nd{\"u}z, ``Optimal utility-privacy trade-off with total variation distance as a privacy measure,'' \emph{IEEE Transactions on Information Forensics and Security}, vol.~15, pp. 594--603, 2019.

\bibitem[Zamani et~al.(2021{\natexlab{a}})Zamani, Oechtering, and Skoglund]{zamani2021data}
A.~Zamani, T.~J. Oechtering, and M.~Skoglund, ``Data disclosure with non-zero leakage and non-invertible leakage matrix,'' \emph{IEEE Transactions on Information Forensics and Security}, vol.~17, pp. 165--179, 2021.

\bibitem[Zamani et~al.(2021{\natexlab{b}})Zamani, Oechtering, and Skoglund]{zamani2021design}
------, ``A design framework for strongly $\chi$$^2$-private data disclosure,'' \emph{IEEE Transactions on Information Forensics and Security}, vol.~16, pp. 2312--2325, 2021.

\bibitem[Zamani et~al.(2023)Zamani, Oechtering, and Skoglund]{zamani2023privacy}
------, ``On the privacy-utility trade-off with and without direct access to the private data,'' \emph{IEEE Transactions on Information Theory}, 2023.

\bibitem[Asoodeh et~al.(2018)Asoodeh, Diaz, Alajaji, and Linder]{asoodeh2018estimation}
S.~Asoodeh, M.~Diaz, F.~Alajaji, and T.~Linder, ``Estimation efficiency under privacy constraints,'' \emph{IEEE Transactions on Information Theory}, vol.~65, no.~3, pp. 1512--1534, 2018.

\bibitem[Massey(1994)]{massey1994guessing}
J.~L. Massey, ``Guessing and entropy,'' in \emph{IEEE International Symposium on Information Theory}, 1994, p. 204.

\bibitem[Malone and Sullivan(2004)]{malone2004guesswork}
D.~Malone and W.~G. Sullivan, ``Guesswork and entropy,'' \emph{IEEE Transactions on Information Theory}, vol.~50, no.~3, pp. 525--526, 2004.

\bibitem[du~Pin~Calmon and Fawaz(2012)]{6483382}
F.~du~Pin~Calmon and N.~Fawaz, ``Privacy against statistical inference,'' in \emph{50th Annual Allerton Conference on Communication, Control, and Computing (Allerton)}, 2012, pp. 1401--1408.

\bibitem[Jiang et~al.(2021)Jiang, Seif, Tandon, and Li]{jiang2021LIPcontextaware}
B.~Jiang, M.~Seif, R.~Tandon, and M.~Li, ``Context-aware local information privacy,'' \emph{IEEE Transactions on Information Forensics and Security}, vol.~16, pp. 3694--3708, 2021.

\bibitem[Zarrabian et~al.(2023)Zarrabian, Ding, and Sadeghi]{zarrabian2023lift}
M.~A. Zarrabian, N.~Ding, and P.~Sadeghi, ``On the lift, related privacy measures, and applications to privacy--utility trade-offs,'' \emph{Entropy}, vol.~25, no.~4, p. 679, 2023.

\bibitem[Zarrabian et~al.(2022)Zarrabian, Ding, and Sadeghi]{zarrabian2022asymmetric}
------, ``Asymmetric local information privacy and the watchdog mechanism,'' in \emph{IEEE Information Theory Workshop (ITW)}, 2022, pp. 7--12.

\bibitem[Gilani et~al.(2024)Gilani, Kurri, Kosut, and Sankar]{gilani2024unifying}
A.~Gilani, G.~R. Kurri, O.~Kosut, and L.~Sankar, ``Unifying privacy measures via maximal ($\alpha$, $\beta$)-leakage ({M}$\alpha$bel),'' \emph{IEEE Transactions on Information Theory}, vol.~70, no.~6, pp. 4368--4395, 2024.

\bibitem[Kurri et~al.(2023)Kurri, Sankar, and Kosut]{kurri2023operational}
G.~R. Kurri, L.~Sankar, and O.~Kosut, ``An operational approach to information leakage via generalized gain functions,'' \emph{IEEE Transactions on Information Theory}, pp. 1--1, 2023.

\bibitem[Sakib et~al.(2023{\natexlab{a}})Sakib, Amariucai, and Guan]{sakib2023variations}
S.~K. Sakib, G.~T. Amariucai, and Y.~Guan, ``Variations and extensions of information leakage metrics with applications to privacy problems with imperfect statistical information,'' in \emph{2023 IEEE 36th Computer Security Foundations Symposium (CSF)}.\hskip 1em plus 0.5em minus 0.4em\relax IEEE, 2023, pp. 407--422.

\bibitem[Sakib et~al.(2023{\natexlab{b}})Sakib, Amariucai, and Guan]{sakib2023measures}
------, ``Measures of information leakage for incomplete statistical information: Application to a binary privacy mechanism,'' \emph{ACM Transactions on Privacy and Security}, vol.~26, no.~4, pp. 1--31, 2023.

\bibitem[Sakib et~al.(2024)Sakib, Amariucai, and Guan]{sakib2024information}
------, ``Information leakage measures for imperfect statistical information: Application to non-bayesian framework,'' \emph{IEEE Transactions on Information Forensics and Security}, 2024.

\bibitem[Chatzikokolakis et~al.(2023)Chatzikokolakis, Cherubin, Palamidessi, and Troncoso]{chatzikokolakis2023bayes}
K.~Chatzikokolakis, G.~Cherubin, C.~Palamidessi, and C.~Troncoso, ``Bayes security: A not so average metric,'' in \emph{2023 IEEE 36th Computer Security Foundations Symposium (CSF)}.\hskip 1em plus 0.5em minus 0.4em\relax IEEE, 2023, pp. 388--406.

\bibitem[Biswas et~al.(2024)Biswas, Dras, Faustini, Fernandes, McIver, Palamidessi, and Sadeghi]{biswas2024bayes}
S.~Biswas, M.~Dras, P.~Faustini, N.~Fernandes, A.~McIver, C.~Palamidessi, and P.~Sadeghi, ``Bayes' capacity as a measure for reconstruction attacks in federated learning,'' \emph{arXiv preprint arXiv:2406.13569}, 2024.

\bibitem[Jurado et~al.(2023)Jurado, Gonze, Alvim, and Palamidessi]{jurado2023analyzing}
M.~Jurado, R.~G. Gonze, M.~S. Alvim, and C.~Palamidessi, ``Analyzing the shuffle model through the lens of quantitative information flow,'' in \emph{2023 IEEE 36th Computer Security Foundations Symposium (CSF)}.\hskip 1em plus 0.5em minus 0.4em\relax IEEE, 2023, pp. 423--438.

\bibitem[Del~Grosso et~al.(2023)Del~Grosso, Pichler, Palamidessi, and Piantanida]{del2023bounding}
G.~Del~Grosso, G.~Pichler, C.~Palamidessi, and P.~Piantanida, ``Bounding information leakage in machine learning,'' \emph{Neurocomputing}, vol. 534, pp. 1--17, 2023.

\bibitem[Aubinais et~al.(2023)Aubinais, Gassiat, and Piantanida]{aubinais2023fundamental}
E.~Aubinais, E.~Gassiat, and P.~Piantanida, ``Fundamental limits of membership inference attacks on machine learning models,'' \emph{arXiv preprint arXiv:2310.13786}, 2023.

\bibitem[Rodríguez-Gálvez et~al.(2021{\natexlab{a}})Rodríguez-Gálvez, Thobaben, and Skoglund]{9611429}
B.~Rodríguez-Gálvez, R.~Thobaben, and M.~Skoglund, ``A variational approach to privacy and fairness,'' in \emph{2021 IEEE Information Theory Workshop (ITW)}, 2021, pp. 1--6.

\bibitem[Rodríguez-Gálvez et~al.(2021{\natexlab{b}})Rodríguez-Gálvez, Bassi, and Skoglund]{9531956}
B.~Rodríguez-Gálvez, G.~Bassi, and M.~Skoglund, ``Upper bounds on the generalization error of private algorithms for discrete data,'' \emph{IEEE Transactions on Information Theory}, vol.~67, no.~11, pp. 7362--7379, 2021.

\bibitem[Kalantari et~al.(2018)Kalantari, Sankar, and Sarwate]{kalantari2018hamming}
K.~Kalantari, L.~Sankar, and A.~D. Sarwate, ``Robust privacy-utility tradeoffs under differential privacy and hamming distortion,'' \emph{IEEE Transactions on Information Forensics and Security}, vol.~13, no.~11, pp. 2816--2830, 2018.

\bibitem[Wu et~al.(2020)Wu, Wagner, and Suh]{wu2020optimal}
B.~Wu, A.~B. Wagner, and G.~E. Suh, ``Optimal mechanisms under maximal leakage,'' in \emph{IEEE Conference on Communications and Network Security (CNS)}, 2020, pp. 1--6.

\bibitem[Saeidian et~al.(2021)Saeidian, Cervia, Oechtering, and Skoglund]{saeidian2021hamming}
S.~Saeidian, G.~Cervia, T.~J. Oechtering, and M.~Skoglund, ``Optimal maximal leakage-distortion tradeoff,'' in \emph{IEEE Information Theory Workshop (ITW)}, 2021, pp. 1--6.

\bibitem[Hsu et~al.(2019)Hsu, Asoodeh, and Calmon]{hsu2019information}
H.~Hsu, S.~Asoodeh, and F.~P. Calmon, ``Information-theoretic privacy watchdogs,'' in \emph{IEEE International Symposium on Information Theory (ISIT)}, 2019, pp. 552--556.

\bibitem[Sadeghi et~al.(2021)Sadeghi, Ding, and Rakotoarivelo]{sadeghi2021properties}
P.~Sadeghi, N.~Ding, and T.~Rakotoarivelo, ``On properties and optimization of information-theoretic privacy watchdog,'' in \emph{IEEE Information Theory Workshop (ITW)}, 2021, pp. 1--5.

\bibitem[Goseling and Lopuhaä-Zwakenberg(2022)]{goseling2022robust}
J.~Goseling and M.~Lopuhaä-Zwakenberg, ``Robust optimization for local differential privacy,'' in \emph{IEEE International Symposium on Information Theory (ISIT)}, 2022, pp. 1629--1634.

\bibitem[Kifer and Machanavajjhala(2014)]{kiferPufferfishFrameworkMathematical2014}
D.~Kifer and A.~Machanavajjhala, ``Pufferfish: {{A}} framework for mathematical privacy definitions,'' \emph{ACM Transactions on Database Systems}, vol.~39, no.~1, pp. 1--36, Jan. 2014.

\end{thebibliography}

\normalsize
\appendices
\section{Proof of Theorem \ref{thrm:lipschitz}}
\label{app:lipschitzproof}
  Fix two arbitrary distributions $P_X,Q_X \in \mathcal P$. We have that 
    \begin{align}
        &|\epsilon_{\min}(P_X)-\epsilon_{\min}(Q_X)| \\[.5em]&= |\max_y \ell_{P_X}(X\to y) - \max_y \ell_{Q_X}(X\to y)| \\[.5em]&\leq \max_y \,\biggl|\,\log \frac{\max_x P_{Y|X=x}}{\sum_x P_X(x)P_{Y|X=x}(y)}  \\ & \qquad \qquad \qquad - \log \frac{\max_x P_{Y|X=x}}{\sum_x Q_X(x)P_{Y|X=x}(y)}\,\biggr| \\[.7em]&= \max_y \,\biggl|\log \frac{\sum_x Q_X(x)P_{Y|X=x}(y)}{\sum_x P_X(x)P_{Y|X=x}(y)}\,\biggr|. 
    \end{align}
    Now, note that for any $a,b > 0$, $|\log\frac{a}{b}| \leq \frac{|a-b|}{\min\{a,b\}}$. Hence
    \begin{align}
        &\max_y \,\biggl|\log \frac{\sum_x Q_X(x)P_{Y|X=x}(y)}{\sum_x P_X(x)P_{Y|X=x}(y)}\,\biggr| \\&\leq \max_y \frac{|\sum_x P_{Y|X=x}(y) (Q_X(x)-P_X(x))|}{\min\{\sum_x Q_X(x)P_{Y|X=x}(y),\sum_x P_X(x)P_{Y|X=x}(y)\}} \\ &\leq \max_y \frac{\sum_x|P_{Y|X=x}(y)(Q_X(x)-P_X(x))|}{\min\{\sum_x Q_X(x)P_{Y|X=x}(y),\sum_x P_X(x)P_{Y|X=x}(y)\}} \\&\leq \max_y \frac{\max_x P_{Y|X=x}(y)||P_X - Q_X||_1}{\min\{\sum_x Q_X(x)P_{Y|X=x}(y),\sum_x P_X(x)P_{Y|X=x}(y)\}} \\&\leq \exp\biggl(C(P_{Y|X},\mathcal P)\biggr)||P_X-Q_X||_1.  
    \end{align}
    Where the second step follows from the triangle inequality, and the third step follows from $P_{Y|X=x}(y) \leq \max_x P_{Y|X=x}(y)$. The final inequality follows from the definition of the local leakage capacity by
    \begin{align}
        \min&\{\sum_x Q_X(x)P_{Y|X=x}(y),\sum_x P_X(x)P_{Y|X=x}(y)\} \\&\geq \inf_{P_X\in\mathcal P} \sum_x P_X(x)P_{Y|X=x}(y),
    \end{align}
    which holds since both $Q_X$ and $P_X$ were chosen in $\mathcal P$. Since $\max_x P_{Y|X=x}(y)$ is independent of the data distributions $P_X$ and $Q_X$, the statement follows. \qed

\section{Proof of Theorem \ref{thm:l1gamma}}
\label{app:l1gammaproof}
 Let $Q_X(x) = P_X(x) + \omega(x)$. Then we have that $\sum_x \omega(x) = 0$ and $||\omega||_1 \leq \beta$ for any $Q_X \in \mathcal B_\beta(P_X)$. Hence, 
    \begin{align}
          &\exp(D_{\infty}(P_Y||Q_Y)) = \max_y \, \frac{\sum_x P_{Y|X=x}(y) P_X(x)}{\sum_x P_{Y|X=x}(y)(P_X(x) + \omega(x))} \\[.5em]&= \max_y\,\frac{P_Y(y)}{P_Y(y) + \sum_x P_{Y|X=x}(y)\omega(x)} \\[.7em]&\leq \max_y\,\frac{P_Y(y)}{P_Y(y) - \frac{\beta}{2}(\max_x P_{Y|X=x}(y) - \min_x P_{Y|X=x}(y))} \\&\leq \max_y \frac{P_Y(y)}{P_Y(y) - \frac{\beta}{2}P_Y(y)e^{\epsilon}} = \frac{1}{1-\frac{\beta e^{\epsilon}}{2}},
          \label{eq:l1uncertaintytrick}
    \end{align}
    where the last inequality follows from the assumption that $P_{Y|X}$ satisfies $\epsilon$-PML with respect to $P_X$, and the first inequality follows by choosing $\omega$ such that all negative weight is put to the maximum channel entry at $y$, and all positive weight is put to the minimum. If $\min_x P_{Y|X=x}(y) = 0$, the last inequality holds with equality for each $y$, and hence the bound is tight.

    For the case $k=1$, we have $\epsilon < \log \frac{1}{1-p_{\min}}$, and we can hence refine the bound on the minimum elements in the channel columns using \cite[Proposition 1]{10619510} as
    \begin{equation}
        \min_{x\in\mathcal X} P_{Y|X=x}(y) \geq \frac{1-e^{\epsilon}(1-p_{\min})}{p_{\min}}P_Y(y), \quad \forall y\in \mathcal Y,
    \end{equation}
    where $p_{\min} \coloneqq \min_{x\in\mathcal X}P_X(x)$. This yields the bound for the case $k=1$. As shown in \cite{10619510}, the lower bound on the minimum elements in the channel matrix is tight for extremal mechanisms as described in Section \ref{sec:prelim}.
    
    To identify the cases in which equality holds for $k>1$, note that the chosen $\omega(x)$ in the first inequality corresponds to a valid probability distribution in $\mathcal B_\beta(P_X)$, and therefore yields a sharp inequality. To obtain equality in \eqref{eq:Spx_l1bound}, it is therefore enough to have $\min_x P_{Y|X=x}(y) = 0$ for all $y \in \mathcal Y$.\qed

\section{Proof of Theorem \ref{thm:optimalbinary}}
\label{app:proofpotimalbinary}
   Let $\pi^{(1)} \coloneqq [p+\nicefrac{\beta}{2}, 1-p-\nicefrac{\beta}{2}]$ and $\pi^{(2)} \coloneqq [p-\nicefrac{\beta}{2},1-p+\nicefrac{\beta}{2}]$ denote the two extreme points of the uncertainty set, i.e., $\mathcal P^* = \{\pi^{(1)},\pi^{(2)}\}$. By Corollary \ref{corr:extremalSet}, we can obtain the region of all mechanisms satisfying $\epsilon$-PML for both $\pi^{(1)}$ and $\pi^{(2)}$ by intersecting the fixed-distribution optimization region polytopes for the two distributions as given in \cite[Lemma 1]{grosse2023extremal}. The maximizing solution will then be one of the two non-trivial vertices of this intersection.\footnote{The two trivial vertices are given by $p_{11} = p_{21} \in \{0,1\}$, see \cite{grosse2023extremal} for details.} Let the general binary mechanism be denoted by 
    \begin{equation}
        P_{Y|X} = \begin{bmatrix}
            p_{11} & 1-p_{11} \\
            p_{21} & 1-p_{21}
        \end{bmatrix}.
    \end{equation} We prove this theorem by calculating the intersection point of the two constraints that constitute an extreme point of the intersection of the two optimization regions for $\pi^{(1)}$ and $\pi^{(2)}$, respectively. Without loss of generality, assume that $p_{11} \geq p_{21}$ and consequently $1-p_{11} \leq 1-p_{21}$. For a general data distribution $P_X$, the extremal inequalities then become 
    \begin{equation}
    \label{eq:extremaleqforward}
        \frac{p_{11}}{p_{11}P_X(x_1) + p_{21}P_X(x_2)} \leq e^{\epsilon}
    \end{equation}
    and 
    \begin{equation}
    \label{eq:extremaleqbackward}
        \frac{1-p_{21}}{(1-p_{11})P_X(x_1) + (1-p_{21})P_X(x_2)} \leq e^{\epsilon}.
    \end{equation}
    By the assumption that $p_{11} \geq p_{21}$, we can further conclude that \eqref{eq:extremaleqforward} is achieved with equality for $\pi^{(2)}$, while \eqref{eq:extremaleqbackward} is achieved with equality for $\pi^{(1)}$. For the respective other data distribution, the leakage value is strictly smaller for each equation. Substituting $\pi^{(2)}$ into \eqref{eq:extremaleqbackward} and assuming equality we obtain
    \begin{equation}
    \label{eq:p21funcofp11}
         p_{21} = \biggl(\frac{1-e^{\epsilon}\pi_1^{(2)}}{e^{\epsilon}\pi_2^{(2)}}\biggr)p_{11}.
    \end{equation}
    Further substituting \eqref{eq:p21funcofp11} into \eqref{eq:extremaleqbackward} we have
    \begin{equation}
        p_{11} = \frac{(1-e^{\epsilon})e^{\epsilon}\pi^{(2)}_2}{1-e^{\epsilon}(\pi^{(1)}_2 + \pi_1^{(2)}-e^{\epsilon}(\pi_1^{(2)}\pi_2^{(1)}-\pi_1^{(1)}\pi_2^{(2)}))}.
    \end{equation}
    Making use of the definitions of $\pi^{(1)}$ and $\pi^{(2)}$, this expression simplifies to 
    \begin{equation}
        p_{11} = \frac{e^{\epsilon}(1-p-\nicefrac{\beta}{2})}{1+\beta e^{\epsilon}},
    \end{equation}
    and consequently, we obtain $p_{21}$ from the relation \eqref{eq:p21funcofp11} as
    \begin{equation}
        p_{21} = 1-\frac{e^{\epsilon}(p+\nicefrac{\beta}{2})}{1+\beta e^{\epsilon}},
    \end{equation}
    which completes the proof. \qed

\section{Proof of Proposition \ref{prop:LCCP}}
\label{app:proofLCCP}
Recall that we denote the elements of the mechanism matrix as $P_{Y|X=x_i}(y_j) = p_{ij}$. By definition, $P_{Y|X} \in \mathcal M(\varepsilon,\mathcal P)$ implies that $C(P_{Y|X},\mathcal P) \leq \varepsilon$. Analogously to Appendix \ref{app:l1gammaproof}, let $P_X(x) = \hat P_X(x) + \omega(x)$ for all $x \in \mathcal X$, with $||\omega||_1 \leq \beta$ and $\sum_x\omega(x) = 0$. Then we can write
    \begin{align}
        &C(P_{Y|X},\mathcal P) = C\big([p_{ij}]_{ij},\mathcal P\big) \\&= \log\max_{j\in [N]} \frac{\max_{i\in[N]}p_{ij}}{\sum_{i=1}^N\hat P_X(x_i)p_{ij} + \sum_{i=1}^N \omega(x_i)p_{ij}} \\[2em]
        &=\log\max_{j\in[N]} \frac{\max_{i\in[N]}p_{ij}}{\sum_{i=1}^N\hat P_X(x_i)p_{ij} - \frac{\beta}{2}\Big(p_j^{\max}-p_j^{\min}\Big)},
    \end{align}
    where $p_j^{\max} \coloneqq \max_{i\in[N]}p_{ij}$ and $p_j^{\min} \coloneqq \min_{i\in[N]}p_{ij}$. That is, $C(P_{Y|X},\mathcal P) \leq \varepsilon$ implies that $\forall j \in [N]$,
    \begin{equation}
        \bigg(1+ \frac{\beta e^{\varepsilon}}{2}\bigg)\max_{i\in[N]}p_{ij} \leq e^{\varepsilon}\bigg(\sum_{i=1}^N\hat P_X(x_i)p_{ij} + \frac{\beta}{2}\min_{i\in[N]}p_{ij}\bigg).
    \end{equation}
    The remaining two constraints follow from the requirement that $[p_{ij}]_{ij}$ is a row-stochastic matrix.
\qed

\section{Proof of Propositions \ref{prop:Laplaceleakage}}
\label{app:proofLaplaceLeakage}
 For the presented Laplace mechanism, the PML is given as 
    \begin{align}
        &\ell(X\to y) \\&= \log \,\frac{\max_{x\in\{x_1,x_2\}}\,\exp(-\frac{|y-x|}{b})}{P_X(x_1)\exp(-\frac{|y-x_1|}{b}) + P_X(x_2)\exp(-\frac{|y-x_2|}{b})}.
    \end{align}
    Recall that $\{x_1,x_2\}= \{-1,1\}$. We get a case distinction around $y=0$ for the maximization 
    \begin{equation}
        \max_{x\in \{x_1,x_2\}}\exp(-\frac{|y-x|}{b}) = \begin{cases}
            \exp(-\frac{|y-1|}{b}), \text{ if } y \geq 0. \\ \exp(-\frac{|y+1|}{b}), \text{ if } y<0.
        \end{cases}
    \end{equation}
    First, consider the cases $y<-1$ and $y>1$. In these cases, we can bound the PML as
    \begin{align}
        &\exp \big(\ell(X\to y)\big) \\[.7em]&\leq
            \begin{cases}
                \frac{\exp(-\frac{(y-1)}{b})}{P_X(x_2)\exp(-\frac{(y-1)}{b})+P_X(x_1)\exp(-\frac{(y+1)}{b})}, \text{ if } y >1\\ \frac{\exp(-\frac{(y+1)}{b})}{P_X(x_2)\exp(-\frac{(y-1)}{b})+P_X(x_1)\exp(-\frac{(y+1)}{b})}, \text{ o.w.,}
            \end{cases} \\[.7em]&\leq \begin{cases} \frac{e^{2/b}}{e^{2/b}P_X(x_2)+P_X(x_1)}, \text{ if } y > 1, \\ \frac{e^{2/b}}{e^{2/b}P_X(x_1)+P_X(x_2)}, \text{ o.w.}\end{cases} \label{eq:laplaceex:p_xdepbound}
   \end{align}
Further in the cases $-1 \leq y \leq 1$ we have
    \begin{equation}
        \exp(\ell(X\to y)) \leq \begin{cases}
            \frac{1}{P_X(x_1)e^{-2/b}+P_X(x_2)}, \text{ if } y \geq 0, \\
            \frac{1}{P_X(x_2)e^{-2/b}+P_X(x_1)}, \text{ o/w.}
        \end{cases}
    \end{equation}
    Hence, by picking the maximizing case in \eqref{eq:laplaceex:p_xdepbound}, the leakage is bounded as
    \begin{equation}
        \ell(X\to y) \leq \frac{2}{b} - \log(e^{2/b}p_{\min} + 1-p_{\min})
    \end{equation} 
    for all outcomes $y \in \mathcal Y$. \qed


\end{document}